\documentclass[11pt]{article}

\usepackage{arxiv}

\usepackage[T1]{fontenc}
\usepackage[utf8]{inputenc}
\usepackage{amsmath}
\usepackage{amssymb}
\usepackage{amsthm}
\usepackage{microtype}
\usepackage{booktabs}
\usepackage{braket}
\usepackage{quantikz}
\usepackage{algorithm}
\usepackage{algorithmic}
\usepackage{enumitem}
\usepackage{url}
\usepackage{hyperref}
\definecolor{linknavy}{rgb}{0.10,0.20,0.60}
\hypersetup{colorlinks=true, linkcolor=linknavy,
            citecolor=linknavy, urlcolor=linknavy}

\theoremstyle{plain}
\newtheorem{theorem}{Theorem}[section]
\newtheorem{lemma}[theorem]{Lemma}
\newtheorem{corollary}[theorem]{Corollary}
\theoremstyle{definition}
\newtheorem{remark}[theorem]{Remark}

\newcommand{\Description}[1]{}

\allowdisplaybreaks[1]

\setlist[itemize]{topsep=6pt, itemsep=6pt, parsep=0pt, leftmargin=*}

\usepackage[datamodel=acmdatamodel,style=acmnumeric,%
            backend=biber]{biblatex}
\title{Exact and Optimal Recursive Quantum Search via\\
       Hilbert-Space Decomposition}

\author{%
  John Burke \\
  School of Computer Science and Statistics \\
  Trinity College Dublin \\
  Dublin 2, Ireland \\
  \texttt{burkej15@tcd.ie} \\
  \href{https://orcid.org/0000-0001-7720-8270}
       {\small\texttt{0000-0001-7720-8270}}
  \And
  Ciaran Mc Goldrick \\
  School of Computer Science and Statistics \\
  Trinity College Dublin \\
  Dublin 2, Ireland \\
  \texttt{ciaran.mcgoldrick@tcd.ie} \\
  \href{https://orcid.org/0000-0001-6442-3262}
       {\small\texttt{0000-0001-6442-3262}}
}

\begin{document}

\maketitle


\begin{abstract}
Current approaches to quantum search fail to deeply exploit extant structure in the underlying Hilbert space. Decomposing the search by this structure empowers new strategies and formulations for quantum search and algorithm design.
We present a new decomposition technique acting directly on this structure by recursively decomposing the Hilbert space and constructing the search operator from reflections over the resulting partition. When initial and target states factorise over this partition, dynamics reduce to a single rotation in a two-dimensional plane at each level, with angle given by a scalar recurrence. This recurrence avoids error accumulation from separately bounding success probabilities at each level, yielding an exact state description enabling treatment of the recursion as a whole.
We obtain the target state deterministically and derive oracle and non-oracle costs independently of the search setting.
For unstructured search, our approach
attains the simultaneously optimal $\Theta(\sqrt{N})$ oracle and
non-oracle gate counts. For spatial search on $d$-dimension grids, it recovers the
$O(\sqrt{N})$ time for $d\geq3$ and the
$O\bigl(\sqrt{N}(\log N)^{3/2}\bigr)$ bound of Aaronson and Ambainis
for $d=2$. The exact description of the recursion extends over our decomposition to new subdivision structures and provides a new approach for applying and analysing recursion in quantum algorithm design.
\end{abstract}


\section{Introduction}
\label{sec:introduction}

Unstructured search is one of the few problems for which quantum computers
have a provable advantage over their classical counterparts. Grover's
algorithm~\cite{Grover} finds a marked item among $N$ candidates using
$\Theta(\sqrt{N})$ oracle calls, a quadratic improvement over classical
exhaustive search that is optimal in query
complexity~\cite{Optimal, zalkaOptimal}. Amplitude
amplification~\cite{AmplitudeAmplification} generalises this mechanism,
allowing the success probability of any subroutine preparing a target to be
amplified with a quadratic improvement over classical sampling.
While Grover's algorithm and amplitude amplification act as a single global
process, dividing the task into smaller sub-searches can yield advantages
beyond the standard approaches. Previous work has shown that these
decompositions of quantum search can reduce implementation
costs~\cite{groverOptGates,zhangKorepinDepth,hardwareGrover2}, enable
distributed implementations~\cite{distributedGrover,distributedGroverExact},
and improve quantum search when access to the search space is restricted by
its geometry~\cite{scott_paper}. In this work, we treat the decomposition of
quantum search as a decomposition of the underlying Hilbert space itself, and
show how a recursive search algorithm can be constructed over this
decomposition. This separates the underlying recursive construction from any
particular search setting, while allowing its instantiation to be adapted to
the structure and constraints of the setting in which it is applied.

Given a partition of the underlying Hilbert space, we build the search algorithm recursively from nested reflection operators rather than from nested amplitude-amplification subroutines. These reflections generate rotations on progressively larger subspaces of the partition, with each level tracking the displacement of the component of the initial state within that subspace. When the initial and target states factorise over a common partition, the non-trivial dynamics at each level of recursion reduce to a single two-dimensional invariant plane whose angle is determined by a scalar recurrence.
The recurrence propagates the rotation exactly at every level, giving an exact description of the state throughout the recursion. This avoids the accumulation of error that arises when the success probability is bounded separately at each stage, allowing the recursion to be treated as a whole. This description can then be used to instantiate the algorithm in different search settings, which we demonstrate for unstructured and spatial search. For unstructured search, the construction attains the simultaneously optimal $\Theta(\sqrt{N})$ oracle and non-oracle gate counts of
Bria\'{n}ski et al.~\cite{hardwareGrover2}. For spatial search on the $d$-dimensional grid, it recovers the $O(\sqrt{N})$ running time of Aaronson
and Ambainis for $d\geq3$ and their $O\bigl(\sqrt{N}(\log N)^{3/2}\bigr)$ bound for $d=2$~\cite{scott_paper}. The approach also highlights greater flexibility in how recursion can be structured within quantum search. In particular, the optimal costs described above can be attained with a fixed rate of subdivision at every level, whereas the prior analyses for unstructured search and for $d\geq 3$-dimensional grids are restricted to subdivision rates that increase with recursion depth. More broadly, the exact recursive description makes it possible to analyse recursive quantum search as a whole, rather than through separate bounds at each level, and provides a versatile method for designing recursive quantum algorithms.

\paragraph{An overview of the construction.}
Suppose $\mathcal{H} = \mathcal{H}_m \otimes \cdots \otimes \mathcal{H}_1$
and that the initial state $\ket{\psi}$ and target $\ket{x}$ are product
states over that partition, with local overlaps
$\braket{x_i|\psi_i} = \sin\theta_i$. Denote
$S_x = \mathbb{I} - 2\ket{x}\bra{x}$ for the oracle and let $S_{\psi_i}$
reflect about $\ket{\psi_{i \ldots 1}}$ within the cumulative block of
registers $i \ldots 1$, acting as the identity $\mathbb{I}$ above it. The construction is
the palindromic recursion
\begin{equation}
\label{eq:intro-W}
    W_i = (S_{\psi_i}W_{i-1})^{t_i}\,S_{\psi_i}\,(W_{i-1}S_{\psi_i})^{t_i},
    \qquad W_0 = S_x,
\end{equation}
in which each $W_{i-1}$ is a reflection, so that $S_{\psi_{i}}W_{i-1}$ is a rotation. 
Given the product structure, the two-subspace decomposition of the reflected subspaces~\cite{halmosSubspaces} reduces the non-trivial action to a single invariant plane, spanned by $\ket{x_{m \ldots i+1}}\ket{\psi_{i \ldots 1}}$ and one further
vector. In that plane $S_{\psi_i}W_{i-1}$ rotates through $-2\gamma_i$, with
\begin{equation}
\label{eq:intro-recurrence}
    \sin\gamma_i = \sin\theta_i\,\sin(2t_{i-1}\gamma_{i-1}),
    \qquad \gamma_1 = \theta_1.
\end{equation}
The rotation generated at a level is characterised by two values: the local
overlap $\sin\theta_i$ within its own subspace and the angle
$2t_{i-1}\gamma_{i-1}$ through which the level below carried its rotation.

Given a decomposition of the global space into $m$ subspaces, after $m-1$ levels, the
iterate $S_{\psi_m}W_{m-1}$ rotates through $-2\gamma_m$ in a plane
containing the full initial state $\ket{\psi}$. A vector within that
plane is the factorised state $\ket{x_m} \otimes \ket{\phi_{m-1}}$, which contains the component of the target within $\mathcal{H}_m$, and a known residual state $\ket{\phi_{m-1}}$. Since the rotation angle
$\gamma_m$ is known exactly, that vector can be reached deterministically through generalised phase-reflections~\cite{longExact, hoyerPhases, AmplitudeAmplification}. Furthermore, the residual state $\ket{\phi_{m-1}}$ sits in the subspace one level down at a known
angular distance from that level's own target $\ket{x_{m-1}} \otimes \ket{\phi_{m-2}}$. The construction applies again on the lower subspaces, and a cascade of $m-1$ corrections acting on progressively smaller subspaces each resolve one component of the target, preparing the target state $\ket{x}$ with
unit probability.
Section~\ref{sec:algorithm} describes the scheme in full.

Such a decomposition offers several advantages. The partial reflections $S_{\psi_i}$ act only on the subspaces accumulated up to level $i$, and can therefore be substantially cheaper to implement than the global reflection $S_{\psi}$. Formulating the decomposition directly over the Hilbert space also makes the construction independent of any particular search setting, allowing it to be instantiated whenever the initial and target states admit the required product decomposition. Finally, because the rotation generated at each level is propagated exactly through the recursion, the analysis avoids the accumulation of error that arises from bounding the success probability separately at each stage, significantly simplifying the derivation of properties of the construction.

We show, in Section~\ref{sec:applications}, that both unstructured search with a single target and spatial search for a single marked vertex on a grid admit the required product decomposition. The former partitions the register encoding the candidates into smaller subregisters, while the latter corresponds to a recursive subdivision of the grid encoding the possible positions of the quantum walker. The construction can then be optimised according to the cost relevant to the particular search setting. In unstructured search, this is the non-oracle gate count, while in spatial search it is the number of time steps taken by the walker. Choosing the intermediate iteration counts $\{t_i\}$ appropriately allows the oracle complexity and the cost of the partial reflections to be controlled simultaneously. This allows the same construction to be optimised according to the cost structure of the setting in which it is applied.

\subsection{Related Work}
\label{subsec:related}

We position our approach within the key relevant literature in recursive quantum search, review the
two settings in which we instantiate our approach, and distinguish this work
from our own prior constructions.

\paragraph{Recursive quantum search.}
Recursive quantum search decomposes a global search problem into a hierarchy of sub-searches, where the output of one level provides the input for the next. Such recursions are often built from nested amplitude-amplification subroutines, whose behaviour depends on the success probability of the input algorithm~\cite{AmplitudeAmplification}. Analysing a recursion of these subroutines therefore requires a guarantee on the success probability at each level~\cite{error_bound_1}. For structured search problems, the recursion can be constructed from searches that establish partial candidate solutions, with each stage promoting valid candidates to the next with a bounded success probability~\cite{cerfNested, schrottenStevens}. Even when the problem is unstructured and no partial solution can be established, the search can still be decomposed recursively by partitioning the space encoding the problem. Within each subspace, the overlap between the initial and target states determines the success probability of the restricted search, providing the guarantee needed to compose these searches recursively. Two examples particularly relevant to this work are within unstructured and spatial search. For unstructured
search, Bria\'{n}ski et al.~\cite{hardwareGrover2} partition the
qubit register and recurse with smaller, cheaper operators over the
sub-registers. For spatial search, Aaronson and
Ambainis~\cite{scott_paper} recursively subdivide the grid so that
the searches at each level act only locally. In both cases the recursion replaces global operations with local ones
while retaining the query scaling of the underlying search.

Our work generalises this decomposition of the underlying search setting as an abstract
decomposition of the Hilbert space encoding the problem itself, and
that the search can be constructed recursively within it. Our requirement from Section~\ref{subsec:definitions}
that the initial and target states factorise over a product
partition enables us to guarantee the success probability of our
scheme throughout the recursion, established exactly as a single
angle recurrence. By defining the recursion through a sequence of
reflection operators instead of through amplitude-amplification
subroutines, as in~\cite{scott_paper, hardwareGrover2}, we establish that the rotation generated at each level
propagates through the recursion by an exact amount, and the success
probability accumulates no residual error from bounding the success of
each stage. This makes it possible to reason about the recursion as a whole, rather than controlling the
accumulated error of a sequence of approximate sub-searches. 
Our construction applies to any setting that can be abstracted as
such a decomposable Hilbert space, and we apply it to the two
instantiation settings described above. Our scheme enables greater fidelity and insight into how recursion
operates within these settings. For example, our approach allows us
to show that quantum search can be recursed upon at a fixed rate of
subdivision, where the corresponding prior analyses~\cite{scott_paper, hardwareGrover2} treat an
increasing rate to manage the accumulated error of the subroutines
through the recursion.

\paragraph{Gate-count optimisation for unstructured search.}
Grover search, while optimal in query complexity, is not optimal in
the number of elementary gates required to realise quantum search.
The original algorithm requires $\Theta(\sqrt{N}\log N)$ non-oracle
gates for its diffusers, each a reflection about the uniform
superposition across all $n=\log_2 N$ qubits at $\Omega(n)$
elementary gates per iterate~\cite{mcx1}. This bound was subsequently
reduced to $O(\sqrt{N}\log\log N)$ by replacing the global diffuser
with partial diffusion operators acting on fewer qubits, while
retaining $O(\sqrt{N})$ oracle complexity~\cite{grover_trade-offs_2002},
and later to $O\bigl(\sqrt{N}\log^{(r)}N\bigr)$ for any constant
$r$~\cite{groverOptGates}. Bria\'{n}ski et al.~\cite{hardwareGrover2}
then reduced the non-oracle cost to the optimal $\Theta(\sqrt{N})$
while retaining $\Theta(\sqrt{N})$ oracle complexity. Their construction achieves this through a recursive decomposition of the search, with smaller operators acting on successive sub-registers and the amplitude gain at each level balanced against the cost of implementing those operators.
They also establish a matching trade-off, whereby any algorithm using
$O(\log(1/\epsilon)\sqrt{N})$ non-oracle gates requires at least
$(1+\epsilon)\tfrac{\pi}{4}\sqrt{N}$ queries. Other works explore heuristics for improving the gate complexity,
interleaving the local search operators introduced for quantum partial search~\cite{partialSearch} with full iterates~\cite{zhangKorepinDepth, palindromicDepth}. We
note~\cite{palindromicDepth} in particular, whose recent numerical
optimisations indicate that a palindromic sequence of partial and full
iterates provides the greatest depth reductions, a concept similar to the
palindromic structure of our
iterate from Equation~\eqref{eq:intro-W}. Beyond this shared symmetry, the approaches differ substantially, as
the palindromic sequence in~\cite{palindromicDepth} composes to a single amplitude-amplification
process rather than a recursion on the search.

\paragraph{Spatial search.}
In the spatial setting, the search space is a physical graph and
operations are constrained by its geometry. Since operations act
locally on the walker's position, correlating vertices at distance
$\ell$ requires $\Omega(\ell)$ steps. With travel and oracle queries
carrying a unit cost per step, an efficient search algorithm must
account for both. Grover search in this model costs
$\Theta(N^{1/2+1/d})$ steps on the $d$-dimensional grid of $N$
vertices, and $\Theta(N)$ at $d=2$, its global diffuser fanning out
across the whole grid at every iterate~\cite{benioffRobot}. Aaronson
and Ambainis~\cite{scott_paper} showed that this cost can be improved
upon by recursing on the search, subdividing the grid recursively and
amplifying within the subgrids. By bounding the success of each
sub-search, they show the search completes with constant probability
in $O(\sqrt{N})$ steps for $d\ge3$ and
$O\big(\sqrt{N}(\log N)^{3/2}\big)$ at $d=2$. While these remain the
best known bounds for employing quantum search natively on the grid,
quantum walks have also proven effective in this setting, traversing
the grid directly and searching in $O(\sqrt{N})$ steps for
$d\ge3$~\cite{walks1} and $O\big(\sqrt{N\log N}\big)$ at
$d=2$~\cite{walks2}, achieving a quadratic speedup over the classical
hitting time on general graphs, including with multiple marked
vertices~\cite{apersSpatial}. To our knowledge, whether $O(\sqrt{N})$
is achievable at $d=2$ remains an open question.

\paragraph{Our prior constructions.}
In~\cite{burkeDeterministic} we used recursion to show that
deterministic search can be performed without generalised-phase
operators. The register is partitioned into blocks of two qubits,
fixing the initial overlap so that a single iterate resolves each
block exactly~\cite{exact}, and recursing on these blocks searches
deterministically using $O(N^{\log_2(3)/2})$ queries in total.
In~\cite{burke26} we introduced the idea of recursing through
reflection operators. We used this to show that global reflections
can be removed from quantum search entirely while retaining the
optimal $\Theta(\sqrt{N})$ oracle count, whereas prior constructions
had reduced their number but not eliminated
them~\cite{grover_trade-offs_2002, groverOptGates, hardwareGrover2,
palindromicDepth}. The recursion there takes the same general form
as Equation~\eqref{eq:intro-W}, but is built from different
operators, with the partial diffusers reflecting about individual
registers rather than cumulative blocks.

The analysis of that construction exploits a degeneracy in the
principal angles generated by the recursive operators. Although the
number of associated eigenspaces grows exponentially with the
recursion depth, their principal angles take only two distinct
values. This is sufficient to determine the probability of measuring
the resolved part of the target, but not the remainder of the state,
so mid-circuit measurements are required to extract the information
obtained at each level. By altering the operators used within this work's construction, the
non-trivial dynamics at each level are instead confined to a single
rotation plane, allowing the full state to be determined exactly
throughout the recursion. This exact state description allows us to
prepare the target deterministically and to derive further properties
of the construction, including its non-oracle cost. Moreover, whereas
the previous work was developed specifically for unstructured search,
the present construction is formulated over a decomposition of the
underlying Hilbert space and can therefore be instantiated in other
search settings.

\subsection{Contributions}
\label{sec:results}
\label{subsec:contributions}

\begin{table}[t]
\centering
\small
\setlength{\belowcaptionskip}{8pt}
\caption{The construction instantiated in two settings and compared with
algorithms native to each. For unstructured search, we report the number
of oracle calls and non-oracle gates used to search $N$ items with a
unique target. For spatial search, we report the number of steps used to
find a unique marked vertex of the $d$-dimensional grid in the $C$-local
model of~\cite{scott_paper}. The spatial costs listed for Grover search
are its costs in the grid model as analysed by
Benioff~\cite{benioffRobot}, rather than claims made in~\cite{Grover}.}
\label{tab:comparison}
\begin{tabular*}{\textwidth}{@{\extracolsep{\fill}}lcccc@{}}
\toprule
& \multicolumn{2}{c}{Unstructured search}
& \multicolumn{2}{c}{Spatial search} \\
\cmidrule(lr){2-3}\cmidrule(lr){4-5}
& Oracle calls & Non-oracle gates & $d \geq 3$ & $d = 2$ \\
\midrule
Grover~\cite{Grover, benioffRobot}
& $\Theta(\sqrt{N})$
& $\Theta(\sqrt{N}\log N)$
& $\Theta\bigl(N^{1/2+1/d}\bigr)$
& $\Theta(N)$ \\
Aaronson--Ambainis~\cite{scott_paper}
& --- & ---
& $O(\sqrt{N})$
& $O\bigl(\sqrt{N}\,(\log N)^{3/2}\bigr)$ \\
Bria\'{n}ski et al.~\cite{hardwareGrover2}
& $\Theta(\sqrt{N})$
& $\Theta(\sqrt{N})$
& --- & --- \\
Quantum walks~\cite{walks1, walks2}
& --- & ---
& $O(\sqrt{N})$
& $O\bigl(\sqrt{N\log N}\bigr)$ \\
\addlinespace
\textbf{This work}
& $\Theta(\sqrt{N})$
& $\Theta(\sqrt{N})$
& $\Theta(\sqrt{N})$
& $\Theta\bigl(\sqrt{N}\,(\log N)^{3/2}\bigr)$ \\
\bottomrule
\end{tabular*}
\end{table}

Our central contribution is a novel decomposition technique for quantum search built over a partition of the underlying Hilbert space. The search is constructed recursively from nested reflection operators defined over the decomposition, and we show that this structure allows the dynamics of the recursion to be determined exactly. Because the construction is defined independently of a particular search setting, it can then be instantiated in different settings, as we demonstrate for unstructured and spatial search.

\begin{itemize}
\item \textbf{A reflection-based decomposition of quantum search.} We show how a quantum search algorithm can be built over a decomposition of the Hilbert space encoding the problem. The search operator is constructed recursively from reflections acting on progressively larger subspaces (Eq.~\eqref{eq:intro-W}). Each level rotates the initial state within the subspace on which it acts, and the recursion tracks this displacement rather than the amplitude of the target tracked by nested amplitude amplification~\cite{cerfNested, scott_paper, hardwareGrover2, groverOptGates}. When the initial and target states factorise over a common partition, the non-trivial dynamics at each level reduce to a single rotation whose angle is determined by the recurrence (Eq.~\eqref{eq:intro-recurrence}). This gives an exact description of the state throughout the recursion and allows the target state $\ket{x}$ to be prepared with unit probability (Theorem~\ref{thm:cascade}). This exact characterisation allows the recursion to be treated as a single process rather than through separate success-probability bounds at each stage, simplifying the derivation of properties of the construction, including its oracle and non-oracle costs.

\item \textbf{Instantiation in unstructured and spatial search.}
We demonstrate that the Hilbert spaces encoding unstructured search and spatial search on the $d$-dimensional grid can be decomposed so that our construction applies in both settings. We then optimise the construction for the non-oracle gate count in unstructured search and for the number of time steps taken by the walker in spatial search. Table~\ref{tab:comparison} compares the resulting instantiations with existing constructions in each setting. For unstructured search, the construction attains the simultaneously optimal $\Theta(\sqrt{N})$ oracle and non-oracle gate counts of Bria\'{n}ski et al.~\cite{hardwareGrover2} (Theorem~\ref{thm:combinatorial}). For spatial search, it attains the $O(\sqrt{N})$ running time of Aaronson and Ambainis~\cite{scott_paper} for $d\geq3$ and their $O\bigl(\sqrt{N}(\log N)^{3/2}\bigr)$ bound for $d=2$. The construction also attains these optimal costs under fixed-rate recursive subdivisions, which, for unstructured search and spatial search with $d\geq3$, are not covered by the corresponding prior analyses.

\item \textbf{Improvements over our prior construction.}
The present construction substantially extends our earlier reflection-based recursion~\cite{burke26}. It follows the same broad recursive structure, but uses different reflection operators that confine the non-trivial dynamics at each level to a single rotation plane. This allows the state to be determined exactly throughout the recursion and the target to be prepared deterministically. The construction also removes the $\log\log N$ overhead in the non-oracle gate count, reducing it from $O(\sqrt{N}\log\log N)$ to the optimal $\Theta(\sqrt{N})$. Whereas the earlier construction was developed specifically for unstructured search, the present one is formulated over an abstract decomposition of the Hilbert space and can therefore be instantiated in other search settings. The corrective cascade and spatial-search instantiation are also new.

\end{itemize}

Together, these results show that recursive quantum search can be constructed and analysed directly through a decomposition of the underlying Hilbert space, providing a new route for designing and analysing recursive quantum algorithms.


\section{Model and Problem Statement}
\label{sec:model}

This section outlines the problem class, the operators, the cost and access
model, and the notation.

\subsection{Product Decomposition}
\label{subsec:definitions}

We consider a Hilbert space that factorises as
$\mathcal{H} = \mathcal{H}_m \otimes \mathcal{H}_{m-1} \otimes \cdots
\otimes \mathcal{H}_1$ over $m$ registers, indexed from the innermost
register $1$ to the outermost register $m$, and both the initial state and
the target factorise over that same partition,
\begin{equation*}
    \ket{\psi} = \ket{\psi_m} \otimes \cdots \otimes \ket{\psi_1},
    \qquad
    \ket{x} = \ket{x_m} \otimes \cdots \otimes \ket{x_1},
\end{equation*}
with $\ket{\psi_i}, \ket{x_i} \in \mathcal{H}_i$ and
$\braket{x_i|\psi_i} \ne 0$ at every level $i$.

We define
\begin{equation}
\label{eq:blocks}
    \ket{\psi_{i\ldots j}} = \ket{\psi_{i}}\ket{\psi_{i-1\ldots j}},
    \quad
    \ket{x_{i\ldots j}} = \ket{x_{i}}\ket{x_{i-1\ldots j}},
\end{equation}
as the components of the initial and target states within registers $i$
through $j$, with $\ket{\psi_{i\ldots i}} = \ket{\psi_i}$ and
$\ket{x_{i\ldots i}} = \ket{x_i}$. Let
\begin{equation}
\label{eq:overlaps}
    \braket{x|\psi} = \sin\theta = \prod_{i=1}^m\braket{x_i|\psi_i}
    = \prod_{i=1}^m\sin\theta_i
\end{equation}
where $\sin\theta \in (0,1)$ is the global overlap between $\ket{x}$ and
$\ket{\psi}$, and each $\sin\theta_i$, with $\theta_i \in (0, \pi/2)$, is
the local overlap within register $i$. Without loss of generality we take
$\braket{x_i|\psi_i} = \sin\theta_i > 0$ at every level, absorbing any
phase into $\ket{x_i}$ at the cost of a global phase on $\ket{x}$. We
define the operators
\begin{equation}
\label{eq:reflections}
    S_x = \mathbb{I} - 2\ket{x}\bra{x},
    \quad
    S_{\psi_i} = \mathbb{I}_{m\ldots i+1} \otimes
    (\mathbb{I}_{i\ldots 1} -
    2\ket{\psi_{i\ldots 1}}\bra{\psi_{i\ldots 1}}),
\end{equation}
where $S_x$ is the oracle and $S_{\psi_i}$ the partial diffuser, acting as
the identity $\mathbb{I}$ on registers $m,\ldots,i+1$ and reflecting about
$\ket{\psi_{i\ldots 1}}$ within registers $i,\ldots,1$. The outermost
$S_{\psi_m}$ acts on the entire space.
Finally, the recursive operator $W_i$ is
\begin{equation}
\label{eq:W-def}
    W_i = (S_{\psi_i} W_{i-1})^{t_i}\, S_{\psi_i}\,
    (W_{i-1} S_{\psi_i})^{t_i}, \quad W_0 = S_x,
\end{equation}
where the integer $t_i \ge 1$ counts the repetitions of the iterate
$S_{\psi_i} W_{i-1}$ in each half of the level, for $2t_i$ in total. Figure~\ref{fig:w-decomp} depicts one level. The symmetry
of~\eqref{eq:W-def} about its central $S_{\psi_i}$ makes each $W_i$
itself a reflection, which the following lemma records together with the
subspace it reflects about.

\begin{lemma}[Conjugation form of $W_i$]
\label{lem:conjugation}
Write $P_{\psi_i}$ and $P_{W_i}$ for the projectors onto the subspaces
reflected by $S_{\psi_i}$ and $W_i$, so that
$S_{\psi_i} = \mathbb{I} - 2P_{\psi_i}$ and
$W_i = \mathbb{I} - 2P_{W_i}$, and set
$A_i := (S_{\psi_i}W_{i-1})^{t_i}$. Then, for every $1 \le i \le m$,
\begin{equation}
\label{eq:conjugation}
    (W_{i-1}S_{\psi_i})^{t_i} = A_i^{\dagger}, \qquad
    W_i = A_i\,S_{\psi_i}\,A_i^{\dagger}, \qquad
    P_{W_i} = A_i\,P_{\psi_i}\,A_i^{\dagger}.
\end{equation}
Therefore $W_i$ is a reflection, about the image of
$\mathrm{range}(P_{\psi_i})$ under $A_i$.
\end{lemma}

\begin{proof}
Both $S_{\psi_i}$ and $W_{i-1}$ are unitary involutions, so
$(S_{\psi_i}W_{i-1})^{-1} = W_{i-1}S_{\psi_i}$. Raising to the $t_i$-th power gives the first identity, and
substituting it into~\eqref{eq:W-def} gives the second. Conjugating $S_{\psi_i} = \mathbb{I} - 2P_{\psi_i}$ by the unitary $A_i$
gives $W_i = \mathbb{I} - 2A_i P_{\psi_i} A_i^{\dagger}$, so
$P_{W_i} = A_i P_{\psi_i} A_i^{\dagger}$ as claimed.
\end{proof}

\paragraph{Standing assumptions.}
Let $\gamma_i$ be the angles the recursion yields, defined by
$\gamma_1 = \theta_1$ and
$\gamma_i = \arcsin\bigl(\sin\theta_i \sin(2t_{i-1}\gamma_{i-1})\bigr)$
for $i \ge 2$. Lemma~\ref{lem:recursion} derives this recurrence and identifies
$\gamma_i$ as half the angle through which the level-$i$ iterate
rotates. Throughout, we assume a bound on the schedule,
\begin{equation}
\label{eq:window-assumption}
    w_i := 2t_i\gamma_i \le \frac{\pi}{2}, \qquad i < m,
\end{equation}
imposed alongside $t_i \ge 1$, preventing any level from over-rotating.
Together with $\gamma_1 = \theta_1 > 0$, this makes every $\gamma_i$
strictly positive by induction. We also assume a condition on the partition,
\begin{equation}
\label{eq:partition-assumption}
    \sin\theta_i \le \frac{\sqrt{3}}{2}, \qquad \text{equivalently}
    \qquad \theta_i \le \frac{\pi}{3},
\end{equation}
at every level, so that the local overlap within any single register is
at most $\sqrt{3}/2$. Under this condition we later establish the
deterministic preparation of the target.

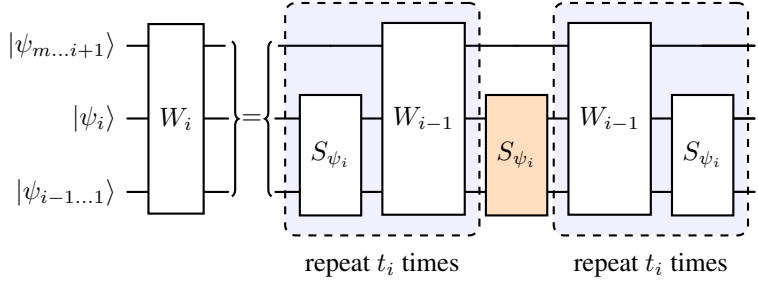
\begin{figure}[t]
\centering
\begin{quantikz}[column sep=0.28cm, row sep=0.35cm]
\lstick{$\ket{\psi_{m\ldots i+1}}$}
  & \gate[3]{W_i}
  & \midstick[3]{$=$}
  & \gategroup[3,steps=2,style={dashed,rounded corners,fill=blue!6,
      inner xsep=2pt},background,
      label style={label position=below,anchor=north,yshift=-0.3cm}]
      {repeat $t_i$ times}
  & \gate[3]{W_{i-1}}
  &
  & \gate[3]{W_{i-1}}\gategroup[3,steps=2,style={dashed,rounded corners,
      fill=blue!6,inner xsep=2pt},background,
      label style={label position=below,anchor=north,yshift=-0.3cm}]
      {repeat $t_i$ times}
  &
  & \qw \\
\lstick{$\ket{\psi_i}$}
  & & & \gate[2]{S_{\psi_i}} &
  & \gate[2,style={fill=orange!25}][0.6cm]{S_{\psi_i}}
  & & \gate[2]{S_{\psi_i}} & \qw \\
\lstick{$\ket{\psi_{i-1\ldots 1}}$}
  & & & & & & & & \qw
\end{quantikz}
\caption{One level of the recursion~\eqref{eq:W-def}, consisting of two
mirrored sequences of $t_i$ alternating applications of $W_{i-1}$ and
$S_{\psi_i}$, separated by a central $S_{\psi_i}$. The reflection
$S_{\psi_i}$ acts on registers $i, \ldots, 1$, while the oracle enters
only at the base through $W_0 = S_x$.}
\label{fig:w-decomp}
\Description{A quantum circuit diagram on three register wires, carrying
from top to bottom the blocks psi at levels m down to i+1, psi at level i,
and psi at levels i-1 down to 1. A single gate W at level i spanning all
three wires is shown to equal a palindromic sequence: a dashed box repeated
t sub i times containing the level i-1 reflection W followed by the partial
diffuser S psi at level i acting only on the lower two wires; then a single
highlighted central copy of that same partial diffuser; then a second dashed
box, again repeated t sub i times, containing the level i-1 reflection and
the partial diffuser in the mirrored order.}
\end{figure}

\subsection{Cost and Access Model}
\label{subsec:access}

The partition, the factors $\ket{\psi_i}$ of the initial state and the
magnitudes $\sin\theta_i$ of the local overlaps are known in advance,
and every parameter of the construction is calculated from these
properties. The construction uses only the oracle, the partial diffusers
$S_{\psi_i}$, and general-phase variants of
both~\cite{longExact, hoyerPhases, AmplitudeAmplification}.

The oracle count $T_{\mathrm{total}}$ is the number of applications of
the oracle by the algorithm. The non-oracle cost is everything the
circuit does between oracle calls, and is defined by the setting the
algorithm is instantiated within. Each partial diffuser is assigned an
abstract cost $c(S_{\psi_i})$, so a setting enters the analysis through
the $m$ numbers $c(S_{\psi_1}), \ldots, c(S_{\psi_m})$. 
The general-phase variants are assigned the same asymptotic costs, since they can be implemented with constant overhead using the corresponding reflection together with an ancilla and single-qubit phase rotations.

\section{The Recursion}
\label{sec:recursion}
\label{sec:algorithm}

This section assembles the operators of Section~\ref{sec:model} into the
full construction (Section~\ref{subsec:construction}). We show that each
iterate $S_{\psi_i}W_{i-1}$ acts as a rotation within a single invariant
plane and derive the recurrence for its angle
(Section~\ref{subsec:reduction}), and show that a cascade of phase-tuned
corrections prepares the target with unit probability
(Section~\ref{subsec:cascade}).

\subsection{The Construction and Proof Overview}
\label{subsec:construction}
The algorithm prepares the target state $\ket{x}$ from the initial
state $\ket{\psi}$. It does so through the recursively defined
operator $W_i$ of Eq.~\eqref{eq:W-def}, which forms the per-level
iterate $S_{\psi_i} W_{i-1}$. The iterate at level $i$ acts non-trivially on a single plane that
contains the state $\ket{x_{m\ldots i+1}} \otimes \ket{\psi_{i\ldots 1}}$. It rotates on this plane through the angle $-2\gamma_i$, a known angle computed in
advance from the local overlaps. The reflection $W_i$ is taken about the
image of the axis of its central $S_{\psi_i}$ under $t_i$ of these
rotations. The angle $2t_i\gamma_i$ carried by the level sets the
angle $\gamma_{i+1}$ of the level above, through the
recurrence~\eqref{eq:intro-recurrence}. This repeats until the outermost
iterate $S_{\psi_m} W_{m-1}$, whose plane contains the full initial
state $\ket{\psi}$ and is not conditioned on any higher subspace
holding its component of $\ket{x}$. We show that this iterate rotates
$\ket{\psi}$ towards $\ket{x_m} \otimes \ket{\phi_{m-1}}$, with
$\ket{\phi_{m-1}}$ a known state of the subspaces beneath, lying in
the plane of the level below at a known angle from that level's own
target $\ket{x_{m-1}} \otimes \ket{\phi_{m-2}}$. The same procedure
therefore applies again, repeating until the state is $\ket{x}$. Every count and phase in this procedure is computed beforehand, in the
classical stage of Algorithm~\ref{alg:protocol}, which states the full
construction.

\renewcommand{\algorithmicrequire}{\textbf{Input:}}
\renewcommand{\algorithmicensure}{\textbf{Output:}}
\begin{algorithm}[t]
\caption{Search via recursive reflections over a decomposed Hilbert space}
\label{alg:protocol}
\begin{algorithmic}[1]
\REQUIRE A partition
$\mathcal{H} = \mathcal{H}_m \otimes \cdots \otimes \mathcal{H}_1$ with
product states $\ket{\psi}$ and $\ket{x}$, local overlaps
$\sin\theta_i$, and a schedule $t_1, \ldots, t_{m-1} \ge 1$ obeying the
bound~\eqref{eq:window-assumption}.
\ENSURE The target
$\ket{x} = \ket{x_m} \otimes \cdots \otimes \ket{x_1}$, up to a global
phase, with unit probability of success.
\item[]
\item[] \textbf{Classical precomputation.} $O(m)$ arithmetic operations,
independent of $\ket{x}$.
\STATE $\gamma_1 \gets \theta_1$
\FOR{$i = 2, \ldots, m$}
    \STATE $\gamma_i \gets
    \arcsin\bigl(\sin\theta_i\,\sin(2t_{i-1}\gamma_{i-1})\bigr)$
    \COMMENT{Lemma~\ref{lem:recursion}}
\ENDFOR
\STATE $J \gets \bigl\lfloor \pi/(4\gamma_m) - 1/2 \bigr\rfloor$
\COMMENT{Lemma~\ref{lem:optimal-tm}}
\STATE $(\alpha_m, \beta_m) \gets$ the phases resolving the residual
rotation $\pi/2 - (2J+1)\gamma_m$
\COMMENT{Lemma~\ref{lem:outermost-step}}
\FOR{$k = m-1, \ldots, 2$}
    \STATE $(\alpha_k, \beta_k) \gets$ the phases
    of~\eqref{eq:fractional-angles} evaluated at $\gamma = \gamma_k$
    \COMMENT{Lemma~\ref{lem:fractional-step}}
\ENDFOR
\STATE $\alpha_1 \gets 2\arcsin\bigl(1/(2\cos\theta_1)\bigr)$, \quad
$\beta_1 \gets -\alpha_1$ \COMMENT{Lemma~\ref{lem:base-step}}
\item[]
\item[] \textbf{Quantum stage 1, resolving the outermost register.}
\STATE Prepare
$\ket{\psi} = \ket{\psi_m} \otimes \cdots \otimes \ket{\psi_1}$
\STATE Apply $(S_{\psi_m}W_{m-1})^{J}$
\COMMENT{Lemma~\ref{lem:optimal-tm}}
\STATE Apply $S_{\psi_m}(\alpha_m)\,W_{m-1}(\beta_m)$
\COMMENT{Lemma~\ref{lem:outermost-step}; state is now
$\ket{x_m} \otimes \ket{\phi_{m-1}}$}
\item[]
\item[] \textbf{Quantum stage 2, the corrective cascade.}
\FOR{$k = m-1, \ldots, 1$}
    \STATE Apply $(W_{k-1}S_{\psi_k})^{t_k}$
    \COMMENT{rewinding the level-$k$ iterates}
    \STATE Apply $S_{\psi_k}(\alpha_k)\,W_{k-1}(\beta_k)$
    \COMMENT{Lemmas~\ref{lem:fractional-step}
    and~\ref{lem:base-step}; state is now
    $\ket{x_{m\ldots k}} \otimes \ket{\phi_{k-1}}$}
\ENDFOR
\RETURN the register contents $\ket{x}$
\COMMENT{Theorem~\ref{thm:cascade}}
\end{algorithmic}
\end{algorithm}

\subsection{Single-Plane Dynamics}
\label{sec:dynamics}
We prove that the non-trivial dynamics of the recursion are confined,
at each level $i$, to the sector
$\ket{x_{m\ldots i+1}} \otimes \mathcal{H}_{i\ldots 1}$, and we
determine those dynamics exactly.
Section~\ref{subsec:reduction} reduces the rotation generated by the
level-$i$ iterate $S_{\psi_i}W_{i-1}$ to a single plane within that
sector (Lemmas~\ref{lem:reduction} and~\ref{lem:eigenvalue}).
Section~\ref{subsec:angle-recurrence} shows that the rotation angle at
each level is determined by the scalar
recurrence~\eqref{eq:engine-recurrence}
(Lemma~\ref{lem:recursion}). Section~\ref{subsec:planes} constructs an
orthonormal basis for the plane and fixes the orientation of the
rotation (Lemma~\ref{lem:basis}).
Section~\ref{subsec:evolution} determines the plane and angle generated
by the outermost iterate $S_{\psi_m}W_{m-1}$ and the overlap of the
evolved state with the target
(Lemmas~\ref{lem:outermost-evolution} and~\ref{lem:optimal-tm}).

\subsubsection{Reduction to One Plane}
\label{subsec:reduction}
Lemma~\ref{lem:reduction} confines the non-trivial action of the
iterate $S_{\psi_i}W_{i-1}$ to the sector
$\ket{x_{m\ldots i+1}} \otimes \mathcal{H}_{i\ldots 1}$, and
Lemma~\ref{lem:eigenvalue} extracts the single principal angle of the
rotation it generates there. Both statements are for the levels
$i \ge 2$, with the base level entering as the case
$\gamma_1 = \theta_1$, where $W_0 = S_x$ makes the iterate ordinary
amplitude amplification about $\ket{x_1}$ on the sector where the
upper subspaces hold $\ket{x_{m\ldots 2}}$.

The rotation angle follows from the two-subspace calculus of
Halmos~\cite{halmosSubspaces}, in which the squared cosines of the
principal angles between two subspaces are the eigenvalues of
$P_1P_2P_1$ on $\mathrm{range}(P_1)$. We write $\Pi_i^\perp$ for the
projector onto the complement of $\ket{x_i}$ within $\mathcal{H}_i$,
\begin{equation}
\label{eq:perp-projector}
    \Pi_i^\perp := \mathbb{I}_i - \ket{x_i}\bra{x_i},
\end{equation}
and $\ket{x_i^\perp}$ for the normalised component of $\ket{\psi_i}$
within that complement,
\begin{equation}
\label{eq:perp-vector}
    \ket{x_i^\perp} := \frac{\Pi_i^\perp\ket{\psi_i}}{\cos\theta_i},
\end{equation}
well defined because
$\bra{\psi_i}\Pi_i^\perp\ket{\psi_i} = \cos^2\theta_i$ is non-zero for
$\theta_i \in (0,\pi/2)$, so that
\begin{equation}
\label{eq:perp-overlap}
    \Pi_i^\perp\ket{\psi_i} = \cos\theta_i\ket{x_i^\perp},
    \qquad
    \braket{x_i^\perp|\psi_i} = \cos\theta_i,
\end{equation}
the second identity following from the first and the positivity of the
overlaps. The pair $\{\ket{x_i},\ket{x_i^\perp}\}$ is then an
orthonormal basis of $\mathrm{span}\{\ket{x_i},\ket{\psi_i}\}$, and
$\ket{\psi_i} = \sin\theta_i\ket{x_i} + \cos\theta_i\ket{x_i^\perp}$.
On multiple subspaces we use only the projector, writing
$\Pi_{m\ldots i}^\perp := \mathbb{I}_{m\ldots i} -
\ket{x_{m\ldots i}}\bra{x_{m\ldots i}}$ and analogously.

\begin{lemma}[Reduction to the relevant subspace]
\label{lem:reduction}
Fix $i \ge 2$ and let $P_\psi$ and $P_W$ denote the projectors onto the
subspaces reflected by $S_{\psi_i}$ and $W_{i-1}$ respectively, so that
$S_{\psi_i} = \mathbb{I} - 2P_\psi$ and $W_{i-1} = \mathbb{I} - 2P_W$.
Write $A_{i-1} := (S_{\psi_{i-1}}W_{i-2})^{t_{i-1}}$, so that
$A_{i-1}^\dagger = (W_{i-2}S_{\psi_{i-1}})^{t_{i-1}}$ by
Lemma~\ref{lem:conjugation}. Then $P_W$ has the block form
\begin{equation}
\label{eq:PW-block}
    P_W = \ket{x_{m\ldots i}}\bra{x_{m\ldots i}} \otimes
    A_{i-1}\ket{\psi_{i-1\ldots 1}}\bra{\psi_{i-1\ldots 1}}A_{i-1}^\dagger
    + \Pi_{m\ldots i}^\perp \otimes
    \ket{\psi_{i-1\ldots 1}}\bra{\psi_{i-1\ldots 1}},
\end{equation}
in whose first term $A_{i-1}$ denotes its restriction to
$\mathcal{H}_{i-1} \otimes \cdots \otimes \mathcal{H}_1$. The subspace
\[
    \mathcal{F} := \{\ket{v}\ket{\psi_i}\ket{\psi_{i-1\ldots 1}} :
    \ket{v} \perp \ket{x_{m\ldots i+1}}\}
\]
is fixed by $S_{\psi_i}W_{i-1}$ with eigenvalue $+1$, and on the sector
$\mathrm{range}(\Pi_{m\ldots i+1}^\perp) \otimes \mathcal{H}_i \otimes
\cdots \otimes \mathcal{H}_1$ the iterate $S_{\psi_i}W_{i-1}$ has
spectrum contained in $\{+1,-1\}$. In the complementary sector the principal
angles between $P_\psi$ and $P_W$ coincide with those between the reduced
projectors
\begin{equation}
\label{eq:reduced-projectors}
\begin{split}
    \tilde{P}_\psi &= \ket{\psi_i}\bra{\psi_i} \otimes
    \ket{\psi_{i-1\ldots 1}}\bra{\psi_{i-1\ldots 1}}
\\
    \tilde{P}_W &= \ket{x_i}\bra{x_i} \otimes
    A_{i-1}\ket{\psi_{i-1\ldots 1}}\bra{\psi_{i-1\ldots 1}}A_{i-1}^\dagger
    + \Pi_i^\perp \otimes \ket{\psi_{i-1\ldots 1}}\bra{\psi_{i-1\ldots 1}}
\end{split}
\end{equation}
acting within $\mathcal{H}_i \otimes \cdots \otimes \mathcal{H}_1$.
\end{lemma}

\begin{proof}
We establish~\eqref{eq:PW-block} by induction on $i$, using
Lemma~\ref{lem:conjugation} in the form
$P_{W_{i-1}} = A_{i-1}P_{\psi_{i-1}}A_{i-1}^\dagger$, and then read the
remaining assertions off it.

At $i = 2$, $P_{\psi_1} = \mathbb{I}_{m\ldots 2} \otimes
\ket{\psi_1}\bra{\psi_1}$ and both factors of $A_1 = (S_{\psi_1}S_x)^{t_1}$
are block-diagonal for the splitting of $\mathcal{H}_{m\ldots 2}$ into
$\ket{x_{m\ldots 2}}$ and its complement, since $S_x = W_0$ restricts to
$\mathbb{I}_1 - 2\ket{x_1}\bra{x_1}$ on the first block and to
$\mathbb{I}_1$ on the second. Hence $A_1$ restricts to
$(S_{\psi_1}W_0)^{t_1}$ on the first block and, on the second, fixes
$\ket{\psi_1}$ up to a sign $(-1)^{t_1}$ that cancels against
$A_1^\dagger$. Conjugating on each block separately
gives~\eqref{eq:PW-block}. For $i \ge 3$ the same argument runs at the
coarser splitting of $\mathcal{H}_{m\ldots i}$.
Assume~\eqref{eq:PW-block} one level down. Both
$\ket{x_{m\ldots i-1}}\bra{x_{m\ldots i-1}}$ and
$\Pi_{m\ldots i-1}^\perp$ are block-diagonal there, and $S_{\psi_{i-1}}$
acts as the identity on $\mathcal{H}_{m\ldots i}$, so $A_{i-1}$ is
block-diagonal too. On the block $\ket{x_{m\ldots i}}$ it restricts to
the operator of~\eqref{eq:PW-block}, and on the complement it restricts
to one fixing $\ket{\psi_{i-1\ldots 1}}$, since $W_{i-2}$ and
$S_{\psi_{i-1}}$ each negate $\ket{\psi_{i-1\ldots 1}}$ there.
Conjugating $P_{\psi_{i-1}}$ on each block completes the induction.

Now $P_\psi = \mathbb{I}_{m\ldots i+1} \otimes
\ket{\psi_{i\ldots 1}}\bra{\psi_{i\ldots 1}}$ has component
$\Pi_{m\ldots i+1}^\perp \otimes \ket{\psi_i}\bra{\psi_i} \otimes
\ket{\psi_{i-1\ldots 1}}\bra{\psi_{i-1\ldots 1}}$ on
$\mathrm{range}(\Pi_{m\ldots i+1}^\perp)$, where the corresponding
component of~\eqref{eq:PW-block} is the larger projector
$\Pi_{m\ldots i+1}^\perp \otimes \mathbb{I}_i \otimes
\ket{\psi_{i-1\ldots 1}}\bra{\psi_{i-1\ldots 1}}$. On this sector the
range of the restriction of $P_\psi$ is therefore contained in that of
$P_W$, so the two projectors commute, the reflections
$S_{\psi_i}$ and $W_{i-1}$ commute with them, and their product is an
involution with spectrum in $\{+1,-1\}$. On $\mathcal{F}$, which is the
range of the restricted $P_\psi$, both projectors act as the identity,
so $S_{\psi_i}W_{i-1}$ applies the phase $(-1)^2 = 1$ there. All
non-real spectrum therefore comes from the $\ket{x_{m\ldots i+1}}$
sector, where the projectors reduce to $\tilde{P}_\psi$ and
$\tilde{P}_W$ as claimed in~\eqref{eq:reduced-projectors}.
\end{proof}
Having confined the rotation to the reduced projectors, we extract the
single principal angle between them as the eigenvalue of
$\tilde{P}_\psi \tilde{P}_W \tilde{P}_\psi$.
\begin{lemma}[Principal angle eigenvalue]
\label{lem:eigenvalue}
The reduced projector $\tilde{P}_\psi$ has rank one, and
$\ket{\psi_{i\ldots 1}}$ is an eigenvector of
$\tilde{P}_\psi \tilde{P}_W \tilde{P}_\psi$ with eigenvalue
\begin{equation}
\label{eq:lambda}
    \lambda = |\braket{x_i|\psi_i}|^2 \cdot
    \left|\braket{\psi_{i-1\ldots 1} | (S_{\psi_{i-1}}W_{i-2})^{t_{i-1}} |
    \psi_{i-1\ldots 1}}\right|^2 + |\braket{x_i^\perp|\psi_i}|^2.
\end{equation}
The single principal angle $\sigma$ between the two subspaces is therefore
$\sigma = \arccos(\sqrt{\lambda}) = |\gamma_i|$, where $\gamma_i$ is the
signed angle of the recurrence~\eqref{eq:engine-recurrence} whose sign
determines the orientation of the rotation plane (Lemma~\ref{lem:basis}). 
\end{lemma}

\begin{proof}
Let $\ket{\Psi} := \ket{\psi_{i-1\ldots 1}}$ and $A := A_{i-1}$, so
that $\tilde{P}_\psi = \ket{\psi_i}\bra{\psi_i} \otimes \ket{\Psi}\bra{\Psi}$
has rank one. Multiplying $\tilde{P}_\psi$ into $\tilde{P}_W$ from the left
register by register, and using
$\bra{\psi_i}\Pi_i^\perp = \braket{\psi_i|x_i^\perp}\bra{x_i^\perp}$
from~\eqref{eq:perp-overlap} together with
$\bra{\Psi}A^{\dagger}\ket{\Psi} = \overline{\braket{\Psi|A|\Psi}}$,
\begin{equation}
\begin{split}
    \tilde{P}_\psi \tilde{P}_W \tilde{P}_\psi
    &= \Big[\braket{\psi_i|x_i}\,\ket{\psi_i}\bra{x_i} \otimes
    \braket{\Psi|A|\Psi}\,\ket{\Psi}\bra{\Psi}A^{\dagger} \\
    &\qquad + \braket{\psi_i|x_i^{\perp}}\,\ket{\psi_i}\bra{x_i^{\perp}}
    \otimes \ket{\Psi}\bra{\Psi}\Big]\tilde{P}_\psi \\
    &= \Big[\,\big|\braket{x_i|\psi_i}\big|^{2}\,
    \big|\braket{\Psi|A|\Psi}\big|^{2} +
    \big|\braket{x_i^{\perp}|\psi_i}\big|^{2}\Big]\tilde{P}_\psi .
\end{split}
\end{equation}
Since $\tilde{P}_\psi = \ket{\psi_{i\ldots 1}}\bra{\psi_{i\ldots 1}}$, this
says that $\ket{\psi_{i\ldots 1}}$ is an eigenvector with the
eigenvalue~\eqref{eq:lambda}, the squared cosine of the unique principal
angle between the two ranges.
\end{proof}

\subsubsection{The Angle Recurrence}
\label{subsec:angle-recurrence}

The dependence of the angle $\gamma_i$ of Lemma~\ref{lem:eigenvalue} on the
operators one level below collapses to a single scalar recurrence in
$\gamma_{i-1}$.

\begin{lemma}[Recursive overlap formula]
\label{lem:recursion}
Let $\sin\theta_i = \braket{x_i | \psi_i}$ and define $\gamma_i$ recursively
by $\gamma_1 = \theta_1$ and
\begin{equation}
\label{eq:engine-recurrence}
    \sin\gamma_i = \sin\theta_i\sin(2t_{i-1}\gamma_{i-1}).
\end{equation}
Then for all $i \ge 1$,
\begin{equation}
    \left|\bra{x_{m\ldots i+1}}\bra{\psi_{i\ldots 1}}
    (S_{\psi_i}W_{i-1})^{t_i}
    \ket{x_{m\ldots i+1}}\ket{\psi_{i\ldots 1}}\right|^2
    = \cos^2(2t_i \gamma_i).
\end{equation}
\end{lemma}

\begin{proof}
We argue in two parts. First, on the $\ket{x_{m\ldots i+1}}$ sector of
Lemma~\ref{lem:reduction}, $S_{\psi_i}$ acts as the rank-one reflection
$\mathbb{I} - 2\tilde{P}_\psi$, so the iterate has a single invariant
plane there and rotates within it through $2\gamma_i$ in magnitude,
$\gamma_i$ being the principal angle of Lemma~\ref{lem:eigenvalue}.
Applied $t_i$ times to the plane vector
$\ket{x_{m\ldots i+1}}\ket{\psi_{i\ldots 1}}$, a rotation through
$\pm 2\gamma_i$ per iterate yields
\begin{equation}
    (S_{\psi_i}W_{i-1})^{t_i}\ket{x_{m\ldots i+1}}\ket{\psi_{i\ldots 1}}
    = \zeta\,\ket{x_{m\ldots i+1}}
    \big(\cos(2t_i\gamma_i)\ket{\psi_{i\ldots 1}} -
    \sin(2t_i\gamma_i)\ket{\psi_{i\ldots 1}^\perp}\big),
\end{equation}
for some unit vector
$\ket{\psi_{i\ldots 1}^\perp} \perp \ket{\psi_{i\ldots 1}}$ in the
plane and a sign $\zeta \in \{+1,-1\}$. Only the modulus of the
overlap is used, on which neither $\zeta$ nor the direction of the
rotation bears, so the overlap equals $\cos^2(2t_i\gamma_i)$.

Second, that this principal angle is the one defined by the recurrence
follows by induction on $i$. At $i = 1$ the iterate is ordinary amplitude
amplification within $\mathcal{H}_1$ on the sector where the upper
subspaces resolve to $\ket{x_{m\ldots 2}}$, so the overlap is
$|\braket{\psi_1 | (S_{\psi_1}S_x)^{t_1} | \psi_1}|^2 = \cos^2(2t_1\theta_1)$,
the claimed formula with $\gamma_1 = \theta_1$. Assuming the formula one
level down, substitution into the eigenvalue~\eqref{eq:lambda} of
Lemma~\ref{lem:eigenvalue} gives
\begin{equation}
    \lambda = \sin^2\theta_i\cos^2(2t_{i-1}\gamma_{i-1}) + \cos^2\theta_i
      = 1 - \sin^2\theta_i\sin^2(2t_{i-1}\gamma_{i-1})
      = 1 - \sin^2\gamma_i = \cos^2\gamma_i
\end{equation}
by the recurrence~\eqref{eq:engine-recurrence}, so
$\arccos(\sqrt{\lambda}) = |\gamma_i|$. The principal angle of
Lemma~\ref{lem:eigenvalue} therefore agrees in magnitude with the
recurrence-defined $\gamma_i$. As per assumption~\eqref{eq:window-assumption}, this keeps it strictly positive.
With the first part, this completes the induction.
\end{proof}

\subsubsection{The Basis of the Plane}
\label{subsec:planes}

The iterate therefore acts as a rotation through $2\gamma_i$ in magnitude
within a two-dimensional invariant plane restricted to
$\ket{x_{m\ldots i+1}} \otimes \mathcal{H}_{i\ldots 1}$. We construct an
orthonormal basis
for that plane, oriented so that the iterate rotates by $-2\gamma_i$ at
level $i$ (Lemma~\ref{lem:basis}).

\begin{lemma}[Orthonormal basis for the rotation plane]
\label{lem:basis}
Define $\ket{\psi_{1}^\perp} := \big(\sin\theta_1\ket{\psi_1} -
\ket{x_1}\big)/\cos\theta_1$ and, recursively for $i \ge 2$,
\begin{equation}
\begin{split}
    \ket{\psi_{i\ldots 1}^\perp} =
    \frac{1}{\cos\gamma_i}\Big[\sin\gamma_i\ket{\psi_i}\ket{\psi_{i-1\ldots 1}}
    - \ket{x_i}\otimes\big(&\sin(2t_{i-1}\gamma_{i-1})
    \ket{\psi_{i-1\ldots 1}} \\
    &+ \cos(2t_{i-1}\gamma_{i-1})\ket{\psi_{i-1\ldots 1}^\perp}\big)\Big].
\end{split}
\label{eq:e2}
\end{equation}
Then, for $i \ge 2$, within the subspace $\ket{x_{m\ldots i+1}} \otimes
\mathcal{H}_i \otimes \cdots \otimes \mathcal{H}_1$ the rotation plane of
$S_{\psi_i}W_{i-1}$ is spanned by $\ket{x_{m\ldots i+1}}\ket{e_1^{(i)}}$
and $\ket{x_{m\ldots i+1}}\ket{e_2^{(i)}}$ with
\begin{equation}
    \ket{e_1^{(i)}} = \ket{\psi_{i\ldots 1}}, \qquad
    \ket{e_2^{(i)}} = \ket{\psi_{i\ldots 1}^\perp},
\end{equation}
and within this plane $S_{\psi_i}W_{i-1}$ acts as a rotation by
$-2\gamma_i$:
\begin{equation}
\label{eq:plane-evolution}
    (S_{\psi_i}W_{i-1})^{t}\ket{\psi_{i\ldots 1}} =
    \cos(2t\gamma_i)\ket{\psi_{i\ldots 1}} -
    \sin(2t\gamma_i)\ket{\psi_{i\ldots 1}^\perp}, \qquad t \in \mathbb{N}.
\end{equation}
At the base level $i = 1$, the same relations hold with $\gamma_1 = \theta_1$
up to an overall sign,
\begin{equation}
\label{eq:base-evolution}
    (S_{\psi_1}W_0)^{t}\ket{\psi_{1}} =
    (-1)^t\big(\cos(2t\theta_1)\ket{\psi_{1}} -
    \sin(2t\theta_1)\ket{\psi_{1}^\perp}\big),
\end{equation}
the factor $(-1)^t$ arising because $W_0 = S_x$ is a rank-one reflection on
the plane rather than the complement of one.
\end{lemma}

\begin{proof}[Proof sketch]
The argument is an induction on $i$. We give its shape here and the
computation in Appendix~\ref{app:proofs}. At the base level both $S_{\psi_1}$
and $S_x$ restrict to rank-one reflections within
$\mathrm{span}\{\ket{x_1},\ket{\psi_1}\}$. The reflected rays $\ket{\psi_1}$
and $\ket{x_1}$ meet at angle $\tfrac{\pi}{2} - \theta_1$, and the rotation twice that angle
$\pi - 2\theta_1$, which is~\eqref{eq:base-evolution}. For $i \ge 2$
the plane is
$\mathrm{span}\{\ket{\psi_{i\ldots 1}}, \tilde{P}_W\ket{\psi_{i\ldots 1}}\}$
by Lemma~\ref{lem:reduction}, so one takes
$\ket{e_1^{(i)}} = \ket{\psi_{i\ldots 1}}$ and obtains $\ket{e_2^{(i)}}$ by
Gram--Schmidt, applying $\tilde{P}_W$ with the level-$(i-1)$
evolution~\eqref{eq:plane-evolution} substituted. The residual vector has
norm $\lvert\sin\gamma_i\cos\gamma_i\rvert$. Dividing by the \emph{signed}
quantity $\sin\gamma_i\cos\gamma_i$ rather than by that norm
yields~\eqref{eq:e2}, with the orientation that makes the iterate rotate at
every level by $-2\gamma_i$, the signed angle
of~\eqref{eq:engine-recurrence}.
\end{proof}

The deferred computation also records the axis of the reflection $W_{i-1}$
within the level-$i$ plane, an identity used repeatedly below and in
Section~\ref{sec:resolving},
\begin{equation}
\label{eq:reflection-axis}
    \ket{w} := \frac{\tilde{P}_W\ket{\psi_{i\ldots 1}}}{\cos\gamma_i}
    = \cos\gamma_i\ket{e_1^{(i)}} + \sin\gamma_i\ket{e_2^{(i)}},
\end{equation}
so that $\tilde{P}_W$ preserves the plane and acts on it as $\ket{w}\bra{w}$.

Rearranging~\eqref{eq:e2} yields the vector carrying the target component
$\ket{x_i}$.

\begin{corollary}[Target ray]
\label{cor:target-ray}
For $1 \le k \le m$ define the \emph{residual state}
\begin{equation}
\label{eq:residual}
    \ket{\phi_{k}} := \sin(2t_{k}\gamma_{k})\ket{\psi_{k\ldots 1}} +
    \cos(2t_{k}\gamma_{k})\ket{\psi_{k\ldots 1}^\perp},
    \qquad \ket{\phi_0} := 1 .
\end{equation}
Then, for every $1 \le k \le m$,
\begin{equation}
\label{eq:target-ray}
    \ket{x_k}\otimes\ket{\phi_{k-1}} = \sin\gamma_k\ket{\psi_{k\ldots 1}} -
    \cos\gamma_k\ket{\psi_{k\ldots 1}^\perp}:
\end{equation}
the product of the level-$k$ target with the residual one level down lies
in the level-$k$ rotation plane, at angle $\gamma_k$ from
$-\ket{\psi_{k\ldots 1}^\perp}$.
\end{corollary}
\begin{proof}
For $k \ge 2$ this is~\eqref{eq:e2} multiplied by $\cos\gamma_k$ and solved
for the $\ket{x_k}$-term. For $k = 1$ it is the definition of
$\ket{\psi_1^\perp}$ solved for $\ket{x_1}$, with $\gamma_1 = \theta_1$ and
$\ket{\phi_0} = 1$.
\end{proof}

These objects all live in the one invariant plane of their level.
Figure~\ref{fig:rotation-plane} draws that plane twice, once for the
outermost evolution from $\ket{e_1^{(m)}} = \ket{\psi}$ and once as the
corrective cascade of Section~\ref{sec:resolving} finds it, at the
residual $\ket{\phi_k}$.

\begin{figure}[t]
\centering
\begin{minipage}[t]{0.48\linewidth}
\centering
\begin{tikzpicture}[thick, scale=0.78]
  \def\R{3.1}
  \draw[gray!40] (0,0) circle (\R);
  \draw[gray!40,->] (-1.12*\R,0) -- (1.12*\R,0);
  \draw[gray!40,->] (0,-1.12*\R) -- (0,1.12*\R);
  \draw[->, very thick] (0,0) -- (0:\R)
      node[right, font=\footnotesize]
      {$\ket{e_1^{(m)}} = \ket{\psi}$};
  \draw[->, very thick] (0,0) -- (90:\R)
      node[above, font=\footnotesize] {$\ket{e_2^{(m)}}$};
  \draw[->, blue!60!black] (0,0) -- (20:\R)
      node[above right, font=\footnotesize, inner sep=1.5pt] {$\ket{w}$};
  \draw[->, red!70!black] (0,0) -- (-70:\R)
      node[below right, font=\footnotesize, inner sep=1.5pt]
      {$\ket{x_m}\otimes\ket{\phi_{m-1}}$};
  \draw[gray!70] (20:{0.30*\R}) -- (-25:{0.4243*\R}) -- (-70:{0.30*\R});
  \draw[->, densely dashed] (0,0) -- (-40:\R)
      node[below right, font=\footnotesize, inner sep=1.5pt, yshift=2mm]
      {$\ket{\psi(t)}$};
  \draw[->, gray!70] (0:{0.42*\R}) arc (0:20:{0.42*\R});
  \node[font=\footnotesize, gray!45!black] at (10:{0.53*\R}) {$\gamma_m$};
  \draw[->, gray!70] (0:{0.68*\R}) arc (0:-40:{0.68*\R});
  \node[font=\footnotesize, gray!45!black] at (-20:{0.80*\R})
      {$-2t\gamma_m$};
  \draw[<->, gray!70, densely dotted]
      (-40:{0.93*\R}) arc (-40:-70:{0.93*\R});
  \node[font=\footnotesize, gray!45!black] at (-55:{1.06*\R}) {$\delta$};
\end{tikzpicture}\\[3pt]
{\small (a) The outermost evolution, from
$\ket{e_1^{(m)}} = \ket{\psi}$.}
\end{minipage}\hfill
\begin{minipage}[t]{0.48\linewidth}
\centering
\begin{tikzpicture}[thick, scale=0.78]
  \def\R{3.1}
  \draw[gray!40] (0,0) circle (\R);
  \draw[gray!40,->] (-1.12*\R,0) -- (1.12*\R,0);
  \draw[gray!40,->] (0,-1.12*\R) -- (0,1.12*\R);
  \draw[->, very thick] (0,0) -- (0:\R)
      node[right, font=\footnotesize]
      {$\ket{e_1^{(k)}} = \ket{\psi_{k\ldots 1}}$};
  \draw[->, very thick] (0,0) -- (90:\R)
      node[above, font=\footnotesize]
      {$\ket{e_2^{(k)}} = \ket{\psi_{k\ldots 1}^\perp}$};
  \draw[->, blue!60!black] (0,0) -- (20:\R)
      node[above right, font=\footnotesize, inner sep=1.5pt] {$\ket{w}$};
  \draw[->, red!70!black] (0,0) -- (-70:\R)
      node[below right, font=\footnotesize, inner sep=1.5pt]
      {$\ket{x_k}\otimes\ket{\phi_{k-1}}$};
  \draw[gray!70] (20:{0.30*\R}) -- (-25:{0.4243*\R}) -- (-70:{0.30*\R});
  \draw[->, green!45!black] (0,0) -- (50:\R)
      node[above right, font=\footnotesize, inner sep=1.5pt, align=left]
      {$\ket{\phi_k}$ \\[-2pt]
       {\scriptsize\color{gray!45!black}$\tfrac{\pi}{2} - 2t_k\gamma_k$}};
  \draw[->, gray!70] (0:{0.42*\R}) arc (0:20:{0.42*\R});
  \node[font=\footnotesize, gray!45!black] at (10:{0.53*\R}) {$\gamma_k$};
  \draw[->, gray!70] (0:{0.86*\R}) arc (0:50:{0.86*\R});
\end{tikzpicture}\\[3pt]
{\small (b) The start of the corrective cascade, at the residual
$\ket{\phi_k}$.}
\end{minipage}
\caption{The invariant plane of Lemma~\ref{lem:basis}, drawn at the
outermost level and at a level of the corrective cascade. In both
panels the iterate is the product of two reflections within the plane.
The first reflection is about $\ket{e_1}$. The second is about the
axis $\ket{w}$ of~\eqref{eq:reflection-axis}, at angle $\gamma$ from
$\ket{e_1}$, so the iterate is the rotation by $-2\gamma$. The target
ray of~\eqref{eq:target-ray} lies at $\gamma - \pi/2$, perpendicular
to $\ket{w}$. In panel~(a) the initial state at level $m$ is
$\ket{e_1^{(m)}} = \ket{\psi}$. Each iterate rotates it by
$-2\gamma_m$ towards the target ray. The phase-tuned step of
Section~\ref{sec:resolving} absorbs the angle $\delta$ left after the
last whole iterate. In panel~(b) level $k$ receives the residual $\ket{\phi_k}$
of~\eqref{eq:residual}, at angle $\pi/2 - 2t_k\gamma_k$. The $t_k$
inverse iterates return it to $\ket{\psi_{k\ldots 1}^\perp}$, and a
single phase-tuned iterate rotates it to the target ray, yielding
$\ket{\phi_{k-1}}$ on the subspaces of level $k-1$.
Angles are drawn at $\gamma = 20^\circ$ and $t = 1$ for legibility.}
\label{fig:rotation-plane}
\Description{Two side-by-side diagrams of the same two-dimensional
invariant plane, each drawn as a unit circle with a horizontal and a
vertical axis. In panel (a) the horizontal axis is the first basis vector
at level m, which is the initial state psi, and the vertical axis is the
second basis vector at level m. The reflection axis w is drawn at angle
gamma sub m above the horizontal axis, and the target ray, the level m
target tensored with the residual one level down, is drawn at gamma sub m
minus ninety degrees, perpendicular to w and marked with a right angle. A
dashed vector, the evolved state after t iterates, lies at minus two t
gamma sub m, part of the way from the horizontal axis towards the target
ray; the angle still separating it from the target ray is labelled delta.
Panel (b) repeats the same construction at a level k of the corrective
cascade, with the axes relabelled accordingly, and adds the residual state
phi sub k above the horizontal axis at ninety degrees minus two t sub k
gamma sub k.}
\end{figure}
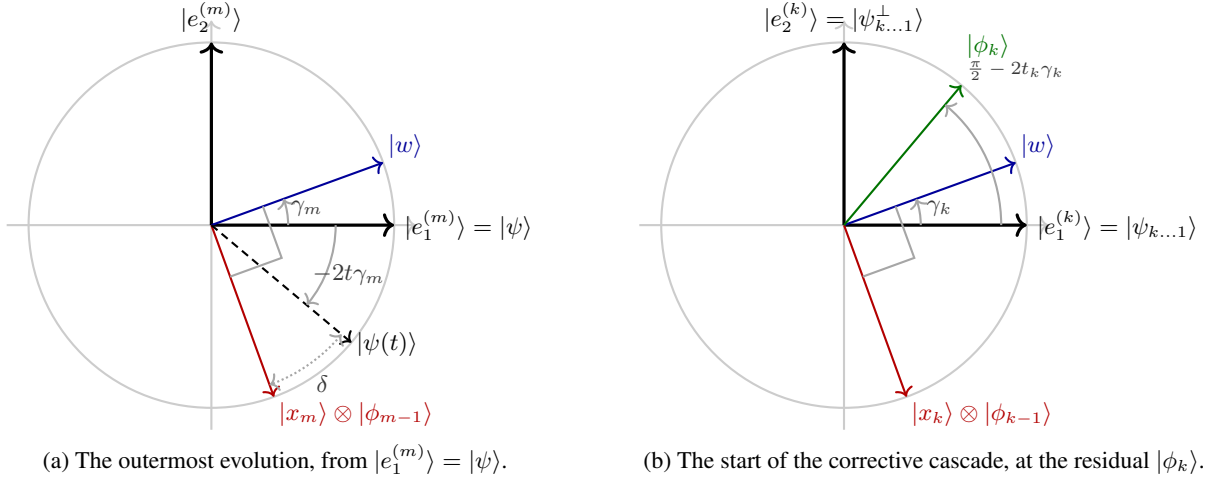

\subsubsection{Outermost Rotation and the Optimal Iteration Count}
\label{subsec:evolution}
At the outermost level $i = m$ the conditioning block
$\ket{x_{m\ldots i+1}}$ is empty and the initial state is
$\ket{\psi} = \ket{e_1^{(m)}}$, so the state remains within the
rotation plane of Lemma~\ref{lem:basis}.
Lemma~\ref{lem:outermost-evolution} gives the state after any number of
applications of the outermost iterate in closed form.
Lemma~\ref{lem:optimal-tm} establishes the iteration count at which the
outermost subspace resolves to $\ket{x_m}$ with certainty.
Remark~\ref{rem:rounding} discusses the probability of success given
the integer rounding of that count, a shortfall the phase-tuned step of
Section~\ref{sec:resolving} removes.

\begin{lemma}[Dynamics at the outermost level]
\label{lem:outermost-evolution}
Let $m \ge 2$. After $t_m$ applications of $S_{\psi_m}W_{m-1}$, the initial
state $\ket{\psi} = \ket{\psi_{m\ldots 1}}$ evolves to
\begin{equation}
    \ket{\psi(t_m)} =
    \frac{1}{\cos\gamma_m}\Big[\cos((2t_m+1)\gamma_m)\ket{\psi_{m\ldots 1}}
    + \sin(2t_m\gamma_m)\ket{x_m}\otimes\ket{\phi_{m-1}}\Big],
\label{eq:psi-tm}
\end{equation}
and for $m = 1$ the same expression holds up to the overall sign
$(-1)^{t_1}$ of~\eqref{eq:base-evolution}.
\end{lemma}

\begin{proof}
Since $\ket{\psi} = \ket{e_1^{(m)}}$ lies entirely within the rotation
plane, \eqref{eq:plane-evolution} gives
$\ket{\psi(t_m)} = \cos(2t_m\gamma_m)\ket{\psi_{m\ldots 1}} -
\sin(2t_m\gamma_m)\ket{\psi_{m\ldots 1}^\perp}$. Now substitute
$\ket{\psi_{m\ldots 1}^\perp} = \big(\sin\gamma_m\ket{\psi_{m\ldots 1}} -
\ket{x_m}\otimes\ket{\phi_{m-1}}\big)/\cos\gamma_m$
from~\eqref{eq:target-ray}. The $\ket{\psi_{m\ldots 1}}$-terms collect with
coefficient $\cos(2t_m\gamma_m) - \tan\gamma_m\sin(2t_m\gamma_m)$, which
equals $\cos((2t_m+1)\gamma_m)/\cos\gamma_m$ by the addition formula
$\cos\gamma_m\cos(2t_m\gamma_m) - \sin\gamma_m\sin(2t_m\gamma_m) =
\cos((2t_m+1)\gamma_m)$. The remaining term is the
$\ket{x_m}\otimes\ket{\phi_{m-1}}$-component shown.
\end{proof}

The right-hand side of~\eqref{eq:psi-tm} is a unit vector of the
rotation plane for every real $t_m$, by~\eqref{eq:target-ray} and the
addition formulae. From here on we treat the count as a real
parameter and write $t \mapsto \ket{\psi(t)}$ for that extension. The
extension locates the count that maximises the probability of
measuring $\ket{x_m}$ within $\mathcal{H}_m$, although the count the
algorithm runs is an integer.

\begin{lemma}[Optimal outermost iteration count]
\label{lem:optimal-tm}
Let $m \ge 2$, let $t \mapsto \ket{\psi(t)}$ be the real-parameter
extension~\eqref{eq:psi-tm}, and let the schedule satisfy the standing
assumption~\eqref{eq:window-assumption}, so that $\gamma_m > 0$. The
probability of measuring $\ket{x_m}$ on the outermost subspace of
$\ket{\psi(t)}$ is
\begin{equation}
\label{eq:prob-tm}
    \Pr[x_m](t) = 1 -
    \frac{\cos^2\theta_m}{\cos^2\gamma_m}\cos^2\big((2t+1)\gamma_m\big),
\end{equation}
which attains its maximum value $1$ exactly at those real $t$ with
$\cos((2t+1)\gamma_m) = 0$. The smallest non-negative such $t$ is
\begin{equation}
\label{eq:tm-star}
    t_m^{*} = \frac{\pi}{4\gamma_m} - \frac{1}{2}.
\end{equation}
At this value the state factorises as
\begin{equation}
    \ket{\psi(t_m^{*})} = \ket{x_m}\otimes\ket{\phi_{m-1}},
    \label{eq:opt-state}
\end{equation}
so a measurement of the outermost subspace of $\ket{\psi(t_m^{*})}$
would yield $\ket{x_m}$ with unit probability. For $m = 1$ the same
statements hold up to the overall sign $(-1)^{t_1}$
of~\eqref{eq:base-evolution}.
\end{lemma}

\begin{proof}
Collect the terms of~\eqref{eq:psi-tm} according to their outermost
subspace. The second already has $\ket{x_m}$ as its outermost factor,
and the first has outermost factor
$\ket{\psi_m} = \sin\theta_m\ket{x_m} + \cos\theta_m\ket{x_m^\perp}$.
The entire component of $\ket{\psi(t)}$ orthogonal to $\ket{x_m}$ on
$\mathcal{H}_m$ is therefore
$\cos\theta_m\cos((2t+1)\gamma_m)/\cos\gamma_m$ times
$\ket{x_m^\perp}\otimes\ket{\psi_{m-1\ldots 1}}$. Subtracting its
squared norm from unity gives~\eqref{eq:prob-tm}. That term is
non-negative and vanishes precisely when
$(2t+1)\gamma_m \in \tfrac{\pi}{2} + \pi\mathbb{Z}$. Since
$\gamma_m > 0$, the smallest non-negative root is~\eqref{eq:tm-star},
on the branch $(2t+1)\gamma_m = \pi/2$. There the first term
of~\eqref{eq:psi-tm} vanishes and
$\sin(2t_m^{*}\gamma_m) = \sin(\tfrac{\pi}{2} - \gamma_m) =
\cos\gamma_m$, so the state collapses to~\eqref{eq:opt-state}.
\end{proof}

\begin{remark}[Rounding]
\label{rem:rounding}
In general $t_m^{*}$ is not an integer, and
$(S_{\psi_m}W_{m-1})^{t_m}$ is defined only at integer $t_m$, so
iteration alone cannot reach the state of~\eqref{eq:opt-state}. Among
integer counts the probability~\eqref{eq:prob-tm} is largest at the
nearest integer $\lfloor t_m^{*}\rceil$, where
$(2\lfloor t_m^{*}\rceil+1)\gamma_m - \tfrac{\pi}{2}$ is at most
$\gamma_m$ in magnitude and therefore
\begin{equation}
    \Pr[x_m] \ge 1 - \cos^2\theta_m\tan^2\gamma_m.
\end{equation}
The phase-tuned construction of Lemma~\ref{lem:outermost-step} removes
the residual error entirely, realising~\eqref{eq:opt-state} exactly
from the integer count $J = \lfloor t_m^{*}\rfloor$. We take the floor
rather than the nearest integer, so the concluding phase-tuned iterate
supplies the rotation that remains instead of undoing an overshoot.
\end{remark}

The recursion is now determined exactly. The state after any number of
iterates is the plane vector~\eqref{eq:psi-tm}, and the count
$J = \lfloor t_m^{*}\rfloor$ brings it within a known angle of
$\ket{x_m}\otimes\ket{\phi_{m-1}}$. Two gaps remain between this point
and the target $\ket{x}$. The final fraction of a rotation at the
outermost level cannot be supplied by whole iterates, and the residual
$\ket{\phi_{m-1}}$ still occupies every subspace beneath the outermost.
Section~\ref{sec:resolving} closes both, with a phase-tuned iterate for
the first and a cascade of such iterates for the second.

\subsection{Resolving the target \texorpdfstring{$\ket{x}$}{|x>}}
\label{sec:resolving}
The rotations that remain are smaller than a whole iterate. We
implement them with iterates whose reflections carry tunable phases in
place of $\pi$, replacing $S_{\psi_k}$ by $S_{\psi_k}(\alpha)$ and
$W_{k-1}$ by $W_{k-1}(\beta)$. The level-$k$ iterate then becomes the
generalised phase reflections of H\o{}yer~\cite{hoyerPhases}, and
Theorem~4 of Brassard et al.~\cite{AmplitudeAmplification} applies
directly. At the base level the phases recover the
zero-theoretical-failure-rate condition of Long~\cite{longExact} as a
special case. A phase-tuned iterate leaves the subspace beneath it at
a known point of the next plane down, at angular distance $\gamma_k$
from that plane's own target ray, and the procedure applies again.

\subsubsection{The Outermost Subspace}
\label{subsec:outermost-step}

By Remark~\ref{rem:rounding}, the shortfall left by
$J = \lfloor t_m^{*}\rfloor$ ordinary iterates is a rotation of less
than $2\gamma_m$ within the outermost plane. One iterate with distinct
tunable phases on its two reflections closes this shortfall exactly.

We define the phase-generalised operators
\begin{equation}
\label{eq:phase-generalised}
S_{\psi_i}(\alpha) = \mathbb{I}_{m\ldots i+1}\otimes\big(\mathbb{I}_{i\ldots 1}
- (1-e^{\mathrm{i}\alpha})\ket{\psi_{i\ldots 1}}\bra{\psi_{i\ldots 1}}\big),
\qquad
W_{i}(\beta) = \mathbb{I} - (1-e^{\mathrm{i}\beta})P_{W_i},
\end{equation}
where $P_{W_i}$ is the projector onto the subspace reflected by $W_i$,
so that $S_{\psi_i}(\pi) = S_{\psi_i}$ and $W_i(\pi) = W_i$. Since
$W_i$ is $S_{\psi_i}$ conjugated by $(S_{\psi_i}W_{i-1})^{t_i}$
(Lemma~\ref{lem:conjugation}), the operator $W_i(\beta)$ is the same
conjugation with the central $S_{\psi_i}$ replaced by
$S_{\psi_i}(\beta)$. At the base level,
$W_0(\beta) = \mathbb{I} - (1-e^{\mathrm{i}\beta})\ket{x}\bra{x}$ is a
generalised-phase oracle.

\begin{lemma}[Resolving $\ket{x_m}$]
\label{lem:outermost-step}
Given $\gamma_m > 0$, by~\eqref{eq:window-assumption}, let
$J = \lfloor t_m^{*}\rfloor$, write $c = \cos\gamma_m$ and
$s = \sin\gamma_m$, and let
\begin{equation}
\label{eq:residual-angle}
    \delta := \frac{\pi}{2} - (2J+1)\gamma_m \in [0, 2\gamma_m)
\end{equation}
be the rotation left after the $J$ ordinary iterates. Then the angles
\begin{equation}
\label{eq:deterministic-angles}
    \cos\beta_m = -\cot(2\gamma_m)\tan\delta, \quad \beta_m \in [0,\pi],
    \qquad
    e^{\mathrm{i}\alpha_m} =
    \frac{(1-e^{\mathrm{i}\beta_m})\,\omega\, s^2 - s\,b}
    {c\,a - (1-e^{\mathrm{i}\beta_m})\,\omega\, c^2},
\end{equation}
with $a = \cos(2J\gamma_m)$, $b = -\sin(2J\gamma_m)$ and
$\omega = ca + sb = \sin\delta$, are real, and the state
\begin{equation}
\label{eq:deterministic-step}
S_{\psi_m}(\alpha_m)\, W_{m-1}(\beta_m)\, (S_{\psi_m}W_{m-1})^{J}\ket{\psi}
\end{equation}
equals the factorised state~\eqref{eq:opt-state} up to a global phase.
The outermost subspace is therefore resolved to $\ket{x_m}$ exactly.
\end{lemma}
\begin{proof}
By Lemma~\ref{lem:basis} the iterate $S_{\psi_m}W_{m-1}$ acts on
$\mathrm{span}\{\ket{e_1^{(m)}},\ket{e_2^{(m)}}\}$ as the product of
the reflections about $\ket{e_1^{(m)}}$ and about the axis
$\ket{w} = c\ket{e_1^{(m)}} + s\ket{e_2^{(m)}}$
of~\eqref{eq:reflection-axis}. Writing
$\ket{\tau} := \ket{x_m}\otimes\ket{\phi_{m-1}} =
s\ket{e_1^{(m)}} - c\ket{e_2^{(m)}}$ for the target ray,
Corollary~\ref{cor:target-ray} gives $\braket{w|\tau} = 0$, so
$\mathbb{I} = \ket{w}\bra{w} + \ket{\tau}\bra{\tau}$ on the plane and
\begin{equation}
\label{eq:plane-duality}
 W_{m-1}(\beta)\big\vert_{\mathrm{plane}} = \mathbb{I} -
 (1-e^{\mathrm{i}\beta})\ket{w}\bra{w}
 = e^{\mathrm{i}\beta}\big[\mathbb{I} -
 (1-e^{-\mathrm{i}\beta})\ket{\tau}\bra{\tau}\big].
\end{equation}
The phase $e^{\mathrm{i}\beta}$ on the axis $\ket{w}$ therefore acts,
up to the global factor $e^{\mathrm{i}\beta}$, as the phase
$e^{-\mathrm{i}\beta}$ on the target $\ket{\tau}$. Restricted to the
plane, the pair $\big(S_{\psi_m}(\alpha), W_{m-1}(\beta)\big)$ is
therefore the phase-generalised amplitude amplification of
H\o{}yer~\cite{hoyerPhases} for the initial state $\ket{\psi}$ and the
good subspace $\mathbb{C}\ket{\tau}$, at amplitude angle $\gamma_m$
because $\braket{\tau|\psi} = s$, and by Theorem~4 of Brassard et
al.~\cite{AmplitudeAmplification} phases exist for which one
generalised iterate after the $J$ ordinary iterates reaches the ray of
$\ket{\tau}$ exactly. It remains to solve for them.

By~\eqref{eq:plane-evolution} the state after $J$ ordinary iterates is
$\ket{\psi(J)} = a\ket{e_1^{(m)}} + b\ket{e_2^{(m)}}$, and the output
of~\eqref{eq:deterministic-step} lands on the ray of $\ket{\tau}$
precisely when its $\ket{w}$-component vanishes. Writing
$\omega := \braket{w|\psi(J)} = ca + sb = \sin\delta$ and
$u := (1-e^{\mathrm{i}\beta_m})\omega$, and using
$S_{\psi_m}(\alpha_m) = \mathrm{diag}(e^{\mathrm{i}\alpha_m}, 1)$ in
the basis $\{\ket{e_1^{(m)}},\ket{e_2^{(m)}}\}$, the vanishing
condition is
\begin{equation}
\label{eq:outer-condition}
 c\,e^{\mathrm{i}\alpha_m}a + s\,b =
 (1-e^{\mathrm{i}\beta_m})\,\omega\,
 \big(c^2 e^{\mathrm{i}\alpha_m} + s^2\big),
 \qquad\text{i.e.}\qquad
 e^{\mathrm{i}\alpha_m} = \frac{u s^2 - sb}{ca - u c^2},
\end{equation}
the second half of~\eqref{eq:deterministic-angles}. The angle
$\alpha_m$ is real exactly when that quotient has modulus one, that
is, when $\lvert u s^2 - sb\rvert^2 = \lvert ca - uc^2\rvert^2$. As
$u$ traverses the circle
$\lvert u - \omega\rvert = \lvert\omega\rvert$ we have
$\lvert u\rvert^2 = 2\omega\,\mathrm{Re}\,u$ with $\omega$ real, so
the condition is linear,
$2cs(sa - cb)\,\mathrm{Re}\,u = \omega(ca - sb)$. Substituting
$\mathrm{Re}\,u = \omega(1 - \cos\beta_m)$ cancels $\omega$, and the
identities $ca - sb = \cos((2J-1)\gamma_m)$,
$sa - cb = \sin((2J+1)\gamma_m)$ and $2cs = \sin 2\gamma_m$,
with~\eqref{eq:residual-angle}, give
$1 - \cos\beta_m =
\sin(2\gamma_m + \delta)/(\sin(2\gamma_m)\cos\delta) =
1 + \cot(2\gamma_m)\tan\delta$, the first half. If $\omega = \sin\delta = 0$, the $J$ iterates already land on the
target ray and any $\beta_m$ satisfies the condition. The value of $\cos\beta_m$ is
admissible throughout, since
$\lvert\cot(2\gamma_m)\rvert\tan\delta \le 1$ on $(0,\pi/2)$. For
$\gamma_m$ below $\pi/4$ the bound holds because
$\delta \le \pi - 2\gamma_m$. A real $\beta_m \in [0,\pi]$ therefore
exists, with $\alpha_m$ read off from~\eqref{eq:outer-condition}. The
output of~\eqref{eq:deterministic-step} is then a unit vector of the
plane with vanishing $\ket{w}$-component, hence $\ket{\tau}$ up to a
global phase, the factorised state~\eqref{eq:opt-state}.
\end{proof}

The outermost step is unconditional because the truncation leaves a
residual $\delta$ anywhere in $[0,2\gamma_m)$, which the two-parameter
family $(\alpha_m,\beta_m)$ absorbs at every $\gamma_m \in (0,\pi/2)$.

\subsubsection{The Corrective Cascade}
\label{subsec:cascade}
The levels beneath level $m$ are resolved within their own rotation
planes. By~\eqref{eq:plane-evolution}, the residual~\eqref{eq:residual}
is the image of $\ket{\psi_{k\ldots 1}^\perp}$ under the $t_k$ iterates
applied within $W_k$,
\begin{equation}
\label{eq:residual-rewind}
    \ket{\phi_{k}} = (S_{\psi_k}W_{k-1})^{t_k}\ket{\psi_{k\ldots 1}^\perp},
\end{equation}
exactly for $k \ge 2$ and up to the sign $(-1)^{t_1}$
of~\eqref{eq:base-evolution} for $k=1$. The state
$\ket{x_k}\otimes\ket{\phi_{k-1}}$ lies in the same plane
by~\eqref{eq:target-ray}, at angular distance $\gamma_k$ from
$\ket{\psi_{k\ldots 1}^\perp}$. Each level is therefore resolved by
$t_k$ inverse iterates followed by one phase-tuned iterate applying
the remaining $\gamma_k$, and the resolution leaves a residual of the
same form one level down.

\begin{lemma}[Corrective rotation]
\label{lem:corrective-rotation}
Let $1 \le k \le m-1$ with $\lvert\gamma_k\rvert \le \pi/3$, and let
$\ket{\phi_k}$ be the residual~\eqref{eq:residual}. Then $t_k$ applications
of the inverse iterate $W_{k-1}S_{\psi_k}$ followed by a single phase-tuned
iterate $S_{\psi_k}(\alpha_k)\,W_{k-1}(\beta_k)$, whose phases are those of
Lemma~\ref{lem:fractional-step} for $k \ge 2$ and of
Lemma~\ref{lem:base-step} for $k = 1$, send
\begin{equation}
    \ket{x_{m\ldots k+1}}\otimes\ket{\phi_{k}} \;\longmapsto\;
    \ket{x_{m\ldots k}}\otimes\ket{\phi_{k-1}},
\end{equation}
up to a global phase.
\end{lemma}

\begin{proof}
Work within the level-$k$ rotation plane on the sector
$\ket{x_{m\ldots k+1}} \otimes \mathcal{H}_k \otimes \cdots \otimes
\mathcal{H}_1$, and write
$v(\xi) := \cos\xi\ket{\psi_{k\ldots 1}} +
\sin\xi\ket{\psi_{k\ldots 1}^\perp}$ for the plane vector at angle
$\xi$. The residual is then $v(\tfrac{\pi}{2} - 2t_k\gamma_k)$
by~\eqref{eq:residual}, and the target is $v(\gamma_k - \tfrac{\pi}{2})$
by~\eqref{eq:target-ray}. By Lemma~\ref{lem:basis} the iterate rotates
the plane by $-2\gamma_k$, so $t_k$ applications of its inverse carry
the residual to $v(\tfrac{\pi}{2}) = \ket{\psi_{k\ldots 1}^\perp}$,
cf.~\eqref{eq:residual-rewind}. The remaining angular distance to the
target ray is $\gamma_k$, which the phase-tuned iterate of
Lemma~\ref{lem:fractional-step} ($k \ge 2$) or
Lemma~\ref{lem:base-step} ($k = 1$) applies.
\end{proof}

That last rotation is half an ordinary iterate, which no integer
number of iterates produces except at $\gamma_k = \pi/3$, where a
single ordinary iterate reaches the target ray. We give the phases
first for $k \ge 2$ and then for the base level, where $W_0 = S_x$
reflects about the target itself.

\begin{lemma}[Fractional corrective step]
\label{lem:fractional-step}
Let $k \ge 2$, write $\gamma=\gamma_{k}$, and suppose
$\gamma\le\pi/3$. Let $\alpha, \beta \in [0, \pi]$ satisfy
\begin{equation}
    \cos\beta = 1 - \frac{1}{2\cos^2\gamma}, \qquad
    e^{\mathrm{i}\alpha} =
    \frac{\cos^2\gamma + \sin^2\gamma\, e^{\mathrm{i}\beta}}
    {\cos^2\gamma\,(1 - e^{\mathrm{i}\beta})},
    \label{eq:fractional-angles}
\end{equation}
the latter of modulus one for $\beta$ and determined by the pair
\begin{equation}
    \cos\alpha = \frac{\cos 2\gamma}{1 + \cos 2\gamma}, \qquad
    \sin\alpha = \frac{\sqrt{1 + 2\cos 2\gamma}}{1 + \cos 2\gamma},
    \label{eq:fractional-branch}
\end{equation}
that is, by $\alpha = \operatorname{atan2}\big(\sqrt{1 + 2\cos 2\gamma},\,
\cos 2\gamma\big) \in [0,\pi]$. Then, on the sector
$\ket{x_{m\ldots k+1}} \otimes \mathcal{H}_k \otimes \cdots \otimes
\mathcal{H}_1$, the iterate $S_{\psi_{k}}(\alpha)\,W_{k-1}(\beta)$
sends $\ket{\psi_{k\ldots 1}^\perp}$ to the target
$\ket{x_{k}}\otimes\ket{\phi_{k-1}} = \sin\gamma\ket{\psi_{k\ldots 1}}
- \cos\gamma\ket{\psi_{k\ldots 1}^\perp}$, up to a global phase.
\end{lemma}

\begin{proof}
Work in the plane basis
$\{\ket{\psi_{k\ldots 1}},\ket{\psi_{k\ldots 1}^\perp}\}$ with
$c=\cos\gamma$, $s=\sin\gamma$. By~\eqref{eq:reflection-axis}, as in
the proof of Lemma~\ref{lem:outermost-step}, the operator
$W_{k-1}(\beta)$ acts on the plane as
$\mathbb{I}-(1-e^{\mathrm{i}\beta})\ket{w}\bra{w}$ with
$\ket{w}=c\ket{\psi_{k\ldots 1}}+s\ket{\psi_{k\ldots 1}^\perp}$, and
$S_{\psi_{k}}(\alpha)$ acts as $\mathrm{diag}(e^{\mathrm{i}\alpha},1)$.
The target
$\ket{x_{k}}\otimes\ket{\phi_{k-1}} = s\ket{\psi_{k\ldots 1}} -
c\ket{\psi_{k\ldots 1}^\perp}$ is orthogonal to $\ket{w}$ by
Corollary~\ref{cor:target-ray}. The output therefore lands on the
target ray iff its $\ket{w}$-component vanishes. For
$\ket{u}=\ket{\psi_{k\ldots 1}^\perp}$, with $\braket{w|u}=s$, that
component is
\begin{equation}
    \braket{w|S_{\psi_{k}}(\alpha)W_{k-1}(\beta)|u}
    = s\big[c^2 + s^2 e^{\mathrm{i}\beta} -
    c^2 e^{\mathrm{i}\alpha}(1-e^{\mathrm{i}\beta})\big],
\end{equation}
so vanishing requires $c^2 e^{\mathrm{i}\alpha}(1-e^{\mathrm{i}\beta})
= c^2 + s^2 e^{\mathrm{i}\beta}$. Taking moduli and using
$|1-e^{\mathrm{i}\beta}|^2=2(1-\cos\beta)$ gives
$c^4-s^4 = 2c^2\cos\beta$, i.e.\ $\cos\beta = 1 - 1/(2c^2)$, which a
real $\beta$ satisfies iff $c^2\ge\tfrac14$, that is, iff
$\gamma\le\pi/3$. The expression for $e^{\mathrm{i}\alpha}$
in~\eqref{eq:fractional-angles} then has modulus one by construction,
and separating its real and imaginary parts returns the
pair~\eqref{eq:fractional-branch}. The second entry of the pair is
non-negative throughout $\gamma\le\pi/3$, so $\alpha \in [0,\pi]$.
With these angles the $\ket{w}$-component vanishes and, the iterate
being unitary and the plane invariant, the output lies on the ray of
$\ket{x_{k}}\otimes\ket{\phi_{k-1}}$.
\end{proof}

\begin{remark}[Computing the phases]
\label{rem:branch}
Since $1 + \cos2\gamma = 2\cos^2\gamma$, the first entry
of~\eqref{eq:fractional-branch} is
$\cos\alpha = 1 - 1/(2\cos^2\gamma) = \cos\beta$, so the two phases
coincide, $\alpha = \beta$, with common value $\pi$ at
$\gamma = \pi/3$. Recovering $\alpha$ from
$\tan\alpha = \sqrt{1+2\cos2\gamma}/\cos2\gamma$ by the principal
arctangent, however, is wrong by exactly $\pi$ for
$\gamma\in(\pi/4,\pi/3]$, and the iterate then misses the target ray
entirely. At $\gamma=\pi/4$, where $\cos2\gamma = 0$, the tangent is
undefined while the pair~\eqref{eq:fractional-branch} gives
$\alpha = \pi/2$ directly.
\end{remark}

At the base level the reflection $W_0 = S_x$ acts about the target
itself rather than through the conjugated projector of the levels
above. The phases form the conjugate pair $\beta = -\alpha$, in place
of the coinciding pair of Remark~\ref{rem:branch}.

\begin{lemma}[Base corrective step]
\label{lem:base-step}
Suppose $\theta_1 \le \pi/3$. Let $\alpha \in [0,\pi]$ satisfy
$\sin(\alpha/2) = 1/(2\cos\theta_1)$ and set $\beta = -\alpha$. Then,
on the sector $\ket{x_{m\ldots 2}} \otimes \mathcal{H}_1$, the iterate
$S_{\psi_1}(\alpha)\,W_0(\beta)$ sends $\ket{\psi_1^\perp}$ to
$\ket{x_1}$ up to a global phase.
\end{lemma}
\begin{proof}
On this sector $W_0(\beta)$ acts as
$\mathbb{I}-(1-e^{\mathrm{i}\beta})\ket{x_1}\bra{x_1}$. Work in
$\mathrm{span}\{\ket{x_1},\ket{\bar x_1}\}$ with
$\ket{\bar x_1} :=
(\ket{\psi_1}-\sin\theta_1\ket{x_1})/\cos\theta_1$, $c = \cos\theta_1$
and $s = \sin\theta_1$, so that
$\ket{\psi_1^\perp} = -c\ket{x_1}+s\ket{\bar x_1}$. For general phases
the $\ket{\bar x_1}$-component of
$S_{\psi_1}(\alpha)W_0(\beta)\ket{\psi_1^\perp}$ is
$s[1 - c^2(1-e^{\mathrm{i}\alpha})(1-e^{\mathrm{i}\beta})]$, so the
output lands on the ray of $\ket{x_1}$ iff
$(1-e^{\mathrm{i}\alpha})(1-e^{\mathrm{i}\beta}) = 1/c^2$. With
$\beta=-\alpha$ the left-hand side is
$|1-e^{\mathrm{i}\alpha}|^2 = 4\sin^2(\alpha/2)$, giving
$\sin(\alpha/2) = 1/(2c)$, which a real $\alpha$ satisfies iff
$c\ge\tfrac12$, that is, iff $\theta_1\le\pi/3$. Unitarity then places
the unit-norm output on $\ket{x_1}$ up to a global phase.
\end{proof}

The condition $\sin(\alpha/2)=1/(2\cos\theta_1)$ is Long's
zero-theoretical-failure-rate condition~\cite{longExact} at zero
preceding iterations, and $(\alpha,-\alpha)$ is an instance of
H\o{}yer's generalised-phase amplitude
amplification~\cite{hoyerPhases}.

\begin{remark}[The generalised-phase oracle]
\label{rem:oracle-free}
Only the base corrective iterate $S_{\psi_1}(\alpha_1)\,W_0(\beta_1)$
of Lemma~\ref{lem:base-step} places a phase other than $\pi$ on the
oracle. At every other level the tuned iterate is
$S_{\psi_k}(\alpha_k)\,W_{k-1}(\beta_k)$ with $k \ge 2$, where
$W_{k-1}(\beta_k) =
A_{k-1}\,S_{\psi_{k-1}}(\beta_k)\,A_{k-1}^{\dagger}$ by
Lemma~\ref{lem:outermost-step}. The tunable phase sits on the central
diffuser, every oracle call inside $A_{k-1}$ and $A_{k-1}^{\dagger}$
is the plain $S_x$, and the same holds for the outermost step whenever
$m \ge 2$. With a black-box oracle giving $S_x$ but not $W_0(\beta)$,
the outermost stage and the corrective steps at levels
$m-1, \ldots, 2$ therefore still conclude exactly, and only the
innermost subspace is left unresolved. Applying the $t_1$ inverse
iterates of level $1$ and stopping leaves
$\ket{x_{m\ldots 2}} \otimes \ket{\psi_1^{\perp}}$, whose measurement
returns $\ket{x}$ with probability $\cos^2\theta_1$. The failure
probability of the black-box variant is thus $\sin^2\theta_1$,
confined to the base level and set by its local overlap alone.
\end{remark}

\subsubsection{The Cascade}
\label{subsec:cascade-theorem}
Applying Lemmas~\ref{lem:corrective-rotation}--\ref{lem:base-step} in
turn, from the outermost resolution of Lemma~\ref{lem:outermost-step},
prepares the full target state.

\begin{theorem}[Exact preparation of the full product target]
\label{thm:cascade}
Let the schedule consist of integers $t_k \ge 1$ for
$1 \le k \le m-1$ satisfying $2t_k\gamma_k \le \pi/2$ and
$\gamma_k \le \pi/3$. Both conditions hold under the standing
assumptions~\eqref{eq:window-assumption}
and~\eqref{eq:partition-assumption}, the second since
$\gamma_k \le \theta_k$. Then the
outermost step of Lemma~\ref{lem:outermost-step} followed by the
corrective steps of Lemma~\ref{lem:corrective-rotation} at levels
$k = m-1, m-2, \ldots, 1$ prepares
\begin{equation}
    \ket{x_m}\otimes\ket{x_{m-1}}\otimes\cdots\otimes\ket{x_1}
\end{equation}
with unit probability, up to a global phase accumulated through the
corrections. Resolving each subspace $\mathcal{H}_k$ with
$1 \le k \le m-1$ requires $t_k$ applications of the inverse iterate
$W_{k-1}S_{\psi_k}$ and one application of the phase-tuned iterate
$S_{\psi_k}(\alpha_k)\,W_{k-1}(\beta_k)$, so $t_k + 1$ applications of
$W_{k-1}$ in total.
\end{theorem}
\begin{proof}
The condition $2t_k\gamma_k \le \pi/2$ makes every $\gamma_k$ strictly
positive, by induction from $\gamma_1 = \theta_1 > 0$. In particular
$\gamma_m > 0$, so $t_m^{*}$ is positive and
$J = \lfloor t_m^{*}\rfloor \ge 0$. Lemma~\ref{lem:outermost-step}
then leaves $\mathcal{H}_m$ in $\ket{x_m}$ and the remaining subspaces
in $\ket{\phi_{m-1}}$. For $k = m-1, \ldots, 1$ in turn,
Lemma~\ref{lem:corrective-rotation} sets $\mathcal{H}_k$ to
$\ket{x_k}$ and leaves $\ket{\phi_{k-1}}$ beneath it, up to a further
global phase. Its condition $\gamma_k \le \pi/3$ holds at every level
by assumption. At the base, where $\gamma_1 = \theta_1$, it is the
condition $\theta_1 \le \pi/3$ of Lemma~\ref{lem:base-step}. Since
$\ket{\phi_0} = 1$ is trivial, the cascade terminates with
$\mathcal{H}_1$ in $\ket{x_1}$.
\end{proof}

The target $\ket{x}$ can thus be prepared deterministically.

\begin{remark}[Omitting the phase-tuned iterates]
\label{rem:no-cascade}
We can establish the probability of measuring $\ket{x}$ without access
to the generalised-phase operators. The protocol applies
$J = \lfloor t_m^{*}\rfloor$ ordinary outermost iterates, then at each
level $k = m-1, \ldots, 1$ only the $t_k$ inverse iterates, and
measures. The success probability obeys
\begin{equation}
\label{eq:phasefree-bound}
    \Pr[x] \;\ge\; 1 - \Big(\delta + \sum_{k=1}^{m-1}\gamma_k\Big)^{2},
\end{equation}
with $\delta \in [0, 2\gamma_m)$ the
residual~\eqref{eq:residual-angle}. At each level of the cascade the
omitted iterate would have closed the single angle $\gamma_k$ to the
target ray, and the outermost stage stops $\delta$ short of its own.
The failure is therefore at most the square of the summed angles. As
we show in Section~\ref{sec:applications}, the schedules there make
the angles decay geometrically, so the sum is dominated by its first
term and $1 - \Pr[x] = O(\sin^2\theta_1)$, independent of $m$ and of
$N$. The generalised-phase iterates therefore yield determinism, while
the scheme without them still succeeds with high probability.
\end{remark}

In this section, we have shown that, given a product partition of the
Hilbert space across which the initial and target states factorise, the
non-trivial dynamics at each level are confined to a single invariant
plane. The scalar recurrence~\eqref{eq:engine-recurrence} determines the
rotation angle in each plane and thereby gives an exact description of
the state for any partition and schedule. Under the standing
assumptions~\eqref{eq:window-assumption}
and~\eqref{eq:partition-assumption}, the phase-tuned outermost step and
corrective cascade prepare the target $\ket{x}$ deterministically. This
exact description across arbitrary product partitions and schedules
makes the construction a practical framework for implementing decomposed
quantum searches and for designing and analysing recursive quantum
algorithms more broadly.

\section{Complexity}
\label{sec:complexity}
We now determine the cost of the construction. We analyse its oracle
complexity and give the conditions under which it remains at the
$\Theta(1/\sin\theta)$ optimum of amplitude amplification. We then
analyse the non-oracle cost and establish when it remains of the order
of the oracle count. We further show that these two conditions can be
combined, and give schedules under which both costs remain optimal.

\subsection{Oracle Complexity}
\label{subsec:oracle-complexity}
We take $\Theta(1/\sin\theta)$ oracle calls as the optimal oracle
complexity~\cite{zalkaOptimal, Optimal} at initial overlap $\sin\theta$, the scaling achieved by
Grover search~\cite{Grover} and by amplitude
amplification~\cite{AmplitudeAmplification}. In this section we show
how the oracle complexity of our decomposition can be kept at this
bound. We count the oracle calls of the full protocol, bound the
decay of the rotation angle during the recursion, and establish the
condition under which the total remains within a constant of
$1/\sin\theta$.

The protocol applies $W_{m-1}$ fewer than $\pi/(4\gamma_m) + 3$
times. The outermost level applies it $J+1 \le t_m^* + 1$ times, the
$J = \lfloor t_m^* \rfloor$ ordinary iterates of
Lemma~\ref{lem:optimal-tm} and the phase-tuned iterate of
Lemma~\ref{lem:outermost-step}, with
$t_m^* = \pi/(4\gamma_m) - \tfrac12$. The cascade applies
$t_k + 1 \le 2t_k$ copies of $W_{k-1}$ at level $k$
(Theorem~\ref{thm:cascade}), with an oracle cost less than or equal to a single iterate of the previous level $W_k$, so the whole cascade amounts to fewer than
two further applications of $W_{m-1}$. Each application makes
$T(W_{m-1}) = \prod_{i=1}^{m-1} 2t_i$ oracle calls
(Eq.~\eqref{eq:W-def}), so
\begin{equation}
\label{eq:total-cost}
    T_{\mathrm{total}} \;<\;
    \Big(\frac{\pi}{4\gamma_m} + 3\Big)\prod_{i=1}^{m-1} 2t_i.
\end{equation}
The product counts the oracle complexity of the outermost iterate and
the angle $\gamma_m$ sets how many iterates are needed. Meeting the
optimum $\Theta(1/\sin\theta)$ requires the two to cancel, with
\begin{equation}
\label{eq:optimality-condition}
    \gamma_m=\Theta\Big(\sin\theta\prod_{i=1}^{m-1}2t_i\Big).
\end{equation}
Each level maps the angle through
$\gamma_{i+1} = \arcsin\!\big(\sin\theta_{i+1}\sin(2t_i\gamma_i)\big)$.
The count $t_i$ increases the angle and the overlap $\sin\theta_{i+1}$
decreases it. Unrolled from $\gamma_1 = \theta_1$, and with
$\sin\theta = \prod_{i}\sin\theta_i$ by~\eqref{eq:blocks},
\begin{equation}
\label{eq:angle-ceiling}
    \sin\gamma_m \;\le\; \sin\theta\prod_{i=1}^{m-1}2t_i.
\end{equation}
The condition~\eqref{eq:optimality-condition} is the matching lower
bound. It remains to show that the relative decay of the angle at
each level does not compound below a constant, which
Lemma~\ref{lem:gamma-bracket} and Theorem~\ref{thm:query-optimal}
establish.

\begin{lemma}[Per-level loss]
\label{lem:gamma-bracket}
Let $w_i := 2t_i\gamma_i$, and suppose $w_i \le \pi/2$ for $i < m$ as
in~\eqref{eq:window-assumption}. Then
\begin{equation}
\label{eq:gamma-bracket}
    \sin\gamma_m \ \ge\ \sin\theta\,\Big(\prod_{j=1}^{m-1}2t_j\Big)
    \prod_{i=1}^{m-1}\Big(1-\frac{w_i^2}{6}\Big).
\end{equation}
\end{lemma}
\begin{proof}
Let $g_i:=\sin\gamma_i$, so the recurrence of
Lemma~\ref{lem:recursion} reads $g_i=\sin\theta_i\sin(w_{i-1})$ with
$g_1=\sin\theta_1$. For $0 \le w \le \pi/2$ we have
$\sin w \ge w(1-w^2/6)$, and $w_{i-1} \ge 2t_{i-1}g_{i-1}$ since
$\gamma_{i-1} \ge \sin\gamma_{i-1}$, so
\begin{equation}
    g_i \ \ge\ 2t_{i-1}\sin\theta_i\,
    g_{i-1}\Big(1-\frac{w_{i-1}^2}{6}\Big).
\end{equation}
Unrolling from $g_1=\sin\theta_1$ and collecting
$\sin\theta_1\prod_{i=2}^{m}\sin\theta_i=\sin\theta$
by~\eqref{eq:blocks} yields~\eqref{eq:gamma-bracket}.
\end{proof}

The oracle count therefore remains optimal whenever
$\prod_i(1-w_i^2/6)$ is bounded below by a constant independent of $m$
and $N$. The following condition on the schedule guarantees this.

\begin{theorem}[Oracle complexity]
\label{thm:query-optimal}
Assume bounded local overlaps
$\theta_i\in[\theta_{\min},\theta_{\max}]\subset(0,\pi/2)$ and that
the schedule satisfies
\begin{equation}
\label{eq:ti-explicit}
    t_i\ \le\ \min\!\Big(\frac{\bar a}{2\sin\theta_{i+1}},
    \ \frac{\pi}{4\gamma_i}\Big)\qquad(i<m)
\end{equation}
for a constant $\bar a<1$. Then
$\sin\gamma_m=\Theta\big(\sin\theta\prod_{i=1}^{m-1}2t_i\big)$ and
$T_{\mathrm{total}}=\Theta(1/\sin\theta)$.
\end{theorem}
\begin{proof}
The second term of~\eqref{eq:ti-explicit} gives $w_i\le\pi/2$, so
Lemma~\ref{lem:gamma-bracket} applies and it suffices to bound
$\prod_{i<m}(1-w_i^2/6)$ below by a constant independent of $m$.
Since $w_i^2/6\le\pi^2/24<\tfrac12$ and $\log(1-x)\ge-2x$ on
$[0,\tfrac12]$, that product is at least
$\exp\big(-\tfrac13\sum_{i<m}w_i^2\big)$, and it remains to bound the
sum. By $\sin(w_{i-1})\le 2t_{i-1}\sin\gamma_{i-1}$ and the first
term of~\eqref{eq:ti-explicit}, the recurrence gives
$\sin\gamma_i\le\bar a\sin\gamma_{i-1}$, hence
\begin{equation}
\label{eq:geometric-decay}
    \sin\gamma_i\ \le\ \sin\theta_1\,\bar a^{\,i-1}.
\end{equation}
Since $\gamma_i\le\theta_i\le\theta_{\max}$, we have
$\gamma_i\le\sin\gamma_i/\cos\theta_{\max}$, and with
$2t_i\le\bar a/\sin\theta_{\min}$,
\begin{equation}
\label{eq:w-decay}
    w_i\ =\ 2t_i\gamma_i
    \ \le\ \frac{\sin\theta_1}{\sin\theta_{\min}\cos\theta_{\max}}
    \;\bar a^{\,i}\ =:\ K\bar a^{\,i}.
\end{equation}
Summing the geometric series,
\begin{equation}
\label{eq:tax-sum}
    \sum_{i<m}w_i^2\ \le\ K^2\,\frac{\bar a^{2}}{1-\bar a^{2}},
\end{equation}
a constant depending only on $\bar a$, $\theta_{\min}$ and
$\theta_{\max}$. Lemma~\ref{lem:gamma-bracket} and the
ceiling~\eqref{eq:angle-ceiling} therefore give
$\sin\gamma_m=\Theta\big(\sin\theta\prod_{i=1}^{m-1}2t_i\big)$, and~\eqref{eq:total-cost} gives
$T_{\mathrm{total}}=\Theta(1/\sin\theta)$.
\end{proof}

The first term of~\eqref{eq:ti-explicit} shows that the count remains
optimal provided the iterations of each level do not exceed the
inverse overlap of the level above,
$2t_i \le \bar a/\sin\theta_{i+1}$. Under this restriction the angle
$\gamma_m$ stays within a constant of~\eqref{eq:angle-ceiling}, and
the total oracle complexity is $\Theta(1/\sin\theta)$. Furthermore, with a fixed overlap $\theta_i = \vartheta$ and a
common count $t$ at every level, the constants behind these
asymptotics can be made explicit.

\begin{corollary}[Oracle excess for fixed overlaps]
\label{cor:epsilon}
Let $\theta_i=\vartheta$ and $t_i=t$ for all $i<m$, with
$a:=2t\sin\vartheta\le\bar a<1$, and set
\begin{equation}
\label{eq:epsilon-def}
    \epsilon\ :=\ \frac{a^2}{6\cos^2\vartheta\,(1-a^2)}
    \ =\ O\big(t^2\sin^2\vartheta\big).
\end{equation}
Then
\begin{equation}
\label{eq:epsilon-uniform}
    \sin\gamma_m\ \ge\ (1-\epsilon)\,\sin\theta\prod_{i=1}^{m-1}2t_i.
\end{equation}
\end{corollary}
\begin{proof}
With a fixed subspace overlap angle the constant of~\eqref{eq:w-decay} is
$K=1/\cos\vartheta$, so~\eqref{eq:tax-sum} gives
$\tfrac16\sum_{i<m}w_i^2\le\epsilon$. Since
$\prod_{i<m}\big(1-\tfrac{w_i^2}{6}\big)\ge1-\tfrac16\sum_{i<m}w_i^2$,
Lemma~\ref{lem:gamma-bracket} gives~\eqref{eq:epsilon-uniform}.
\end{proof}

Substituting~\eqref{eq:epsilon-uniform} into the first term
of~\eqref{eq:total-cost}, using $\sin\gamma_m \le \gamma_m$, for
$\epsilon \le \tfrac12$,
\begin{equation}
\label{eq:explicit-total}
    T_{\mathrm{total}} \;<\;
    (1+2\epsilon)\,\frac{\pi}{4\sin\theta}
    \;+\; 3\prod_{i=1}^{m-1} 2t_i.
\end{equation}
The oracle complexity is therefore $\Theta(1/\sin\theta)$, with
leading constant $\tfrac{\pi}{4}(1+2\epsilon)$ and the second term at
most $6/\sin\theta$ by~\eqref{eq:epsilon-uniform}. The count
$t_i = 1$ minimises both $\epsilon$ and the product, so the minimal
schedule is the cheapest.

The complexity can also be determined at the boundary $\bar a = 1$, where it becomes $\Theta(\sqrt{m}/\sin\theta)$.

\begin{lemma}[Oracle count at the boundary]
\label{lem:boundary}
Let the schedule satisfy $2t_i\sin\theta_{i+1}=1$ and $w_i\le\pi/2$
for $i<m$. Then
\begin{equation}
\label{eq:boundary-decay}
    \gamma_m=\Theta\Big(\big(\textstyle\sum_{i=1}^{m}
    \sin^{-2}\theta_i\big)^{-1/2}\Big),
\end{equation}
with constants uniform in the overlaps. For overlaps inside a fixed
$[\theta_{\min},\theta_{\max}]\subset(0,\pi/2)$ this is
$\gamma_m=\Theta\big(1/\sqrt{m}\,\big)$, and
\begin{equation}
\label{eq:boundary-cost}
    T_{\mathrm{total}}=\Theta\Big(\frac{\sqrt{m}}{\sin\theta}\Big).
\end{equation}
\end{lemma}
\begin{proof}[Proof sketch]
The full computation is given in Appendix~\ref{app:proofs}. At the
boundary $\sin\theta_{i+1}=\gamma_i/w_i$, so the recurrence of
Lemma~\ref{lem:recursion} reduces to
$\sin\gamma_{i+1}=\gamma_i\sin(w_i)/w_i$. We track
$y_i:=1/\gamma_i^{2}$. Standard bounds on $\sin w$ for $w_i\le\pi/2$
give
\begin{equation}
\label{eq:boundary-increment}
    y_{i+1}-y_i=\Theta\Big(\frac{1}{\sin^2\theta_{i+1}}\Big),
\end{equation}
with constants independent of the overlap. Summing from
$y_1=\Theta\big(\sin^{-2}\theta_1\big)$
gives~\eqref{eq:boundary-decay}. Inside a fixed interval every term
of the sum is $\Theta(1)$, so $\gamma_m=\Theta(1/\sqrt{m})$, and the
schedule fixes $\prod_{i=1}^{m-1}2t_i=\sin\theta_1/\sin\theta
=\Theta(1/\sin\theta)$, which substituted
into~\eqref{eq:total-cost} yields~\eqref{eq:boundary-cost}.
\end{proof}

Thus by bounding the per-level iterates as in Theorem~\ref{thm:query-optimal}, the construction attains the optimal oracle complexity $\Theta(1/\sin\theta)$.

\subsection{Non-Oracle Operators}
\label{sec:non-oracle}
We now determine when the cost of the non-oracle operators remains proportional to the construction's oracle complexity. Let $c(S_{\psi_i})$ denote the
implementation cost of the diffuser $S_{\psi_i}$. The diffusers are
commuting involutions, so within each iterate a single diffuser remains
between successive oracle calls
(Lemma~\ref{lem:diffuser-cancellation}), and
Lemma~\ref{lem:nonoracle-cost} bounds the cost of the remaining
operators against the oracle count. The bound stays of the order of the
oracle count when the growth in diffuser cost from one level to the
next is offset by the repetitions introduced at the lower level. Under
this condition the contribution of each successive diffuser decreases
geometrically, and the total non-oracle cost remains proportional to the
oracle count multiplied by the cost of the innermost diffuser
(Corollary~\ref{cor:geometric-decay}).

Before cancellation, the total non-oracle cost of a single application
of the outermost iterate $S_{\psi_m}W_{m-1}$ is
\begin{equation}
\label{eq:naive-count}
    C(S_{\psi_m}W_{m-1}) = \sum_{i=1}^{m-1}c(S_{\psi_i})(2t_i+1)
    \prod_{j=i+1}^{m-1}2t_j+c(S_{\psi_m}),
\end{equation}
and Lemma~\ref{lem:diffuser-cancellation} reduces this count by merging
the diffusers that meet between consecutive oracle calls.

\begin{lemma}[Diffuser cancellation]
\label{lem:diffuser-cancellation}
The diffusion operators $\{S_{\psi_i}\}_{i=1}^{m}$ are pairwise
commuting involutions. Within a single iterate $S_{\psi_m}W_{m-1}$, a
single diffuser therefore remains between each pair of consecutive
oracle calls, and the non-oracle cost reduces to
\begin{equation}
\label{eq:nonoracle-reduced}
    C(S_{\psi_m}W_{m-1})
      = \sum_{i=1}^{m-1} c(S_{\psi_i})\,(2t_i-1)
        \!\!\prod_{j=i+1}^{m-1}\!\! 2t_j
        \;+\; c(S_{\psi_m})+2\!\sum_{i=1}^{m-1}\! c(S_{\psi_i}) .
\end{equation}
\end{lemma}
\begin{proof}
For $j<i$ the axis $\ket{\psi_{i\ldots 1}}$ of $S_{\psi_i}$ contains
$\ket{\psi_{j\ldots 1}}$ as a tensor factor and is therefore a
$(-1)$-eigenvector of $S_{\psi_j}$, so, each diffuser being an
involution, $S_{\psi_j}S_{\psi_i}S_{\psi_j}=S_{\psi_i}$.
By~\eqref{eq:W-def}, consecutive copies of $W_{i-1}$ inside $W_i$
appear as $W_{i-1}S_{\psi_i}W_{i-1}$, and since every $W_{i-1}$ opens
and closes with $S_{\psi_{i-1}}$,
\begin{equation*}
W_{i-1}\,S_{\psi_i}\,W_{i-1}
= \cdots W_{i-2}\,S_{\psi_{i-1}}S_{\psi_i}S_{\psi_{i-1}}\,W_{i-2}\cdots
= \cdots W_{i-2}\,S_{\psi_i}\,W_{i-2}\cdots .
\end{equation*}
Repeating down to $W_0=S_x$ leaves a single $S_{\psi_i}$ between the
two adjacent oracle calls. Each $W_i$ contains $2t_i$ copies of
$W_{i-1}$, hence $2t_i-1$ consecutive pairs, and occurs
$\prod_{j=i+1}^{m-1}2t_j$ times inside $W_{m-1}$, giving the first
term of~\eqref{eq:nonoracle-reduced}. The remaining terms count the
diffusers before the first and after the last oracle call, one
$S_{\psi_i}$ at each end for every level $i<m$ together with the
outermost $S_{\psi_m}$.
\end{proof}

We now compare the remaining cost with the oracle count of the iterate,
$T(W_{m-1}) = \prod_{i=1}^{m-1}2t_i$, weighting the cost of each
diffuser by the oracle calls contributed by the levels beneath it.

\begin{lemma}[Non-oracle cost]
\label{lem:nonoracle-cost}
Define
\begin{equation}
\label{eq:residual-weights}
    u_i := \frac{c(S_{\psi_i})}{\prod_{j=1}^{i-1}2t_j}, \qquad i=1,\ldots,m,
\end{equation}
with $\prod_{j=1}^{0}2t_j=1$. Then
\begin{equation}
\label{eq:nonoracle-theta}
    C(S_{\psi_m}W_{m-1}) = \Theta\!\Big(T(W_{m-1})\sum_{i=1}^{m}u_i\Big),
\end{equation}
so the non-oracle cost is of the order of the oracle count exactly when
$\sum_{i=1}^m u_i=O(1)$.
\end{lemma}
\begin{proof}
Substituting $c(S_{\psi_i})=u_i\prod_{j=1}^{i-1}2t_j$ and
$\prod_{j=1}^{i-1}2t_j\cdot\prod_{j=i+1}^{m-1}2t_j=T(W_{m-1})/2t_i$ into
the first term of \eqref{eq:nonoracle-reduced},
\begin{equation}
\label{eq:nonoracle-residual}
    \sum_{i=1}^{m-1} c(S_{\psi_i})\,(2t_i-1)\!\!\prod_{j=i+1}^{m-1}\!\! 2t_j
      = T(W_{m-1})\sum_{i=1}^{m-1}\Big(1-\frac{1}{2t_i}\Big)u_i ,
\end{equation}
with coefficients $1-\tfrac1{2t_i}\in[\tfrac12,1)$ since $t_i\ge1$. Of the
boundary terms, $c(S_{\psi_m})=u_m\,T(W_{m-1})$ directly from
\eqref{eq:residual-weights}. The leading and trailing runs contribute
$2\sum_{i<m}c(S_{\psi_i})
=2\,T(W_{m-1})\sum_{i<m}u_i/\prod_{j=i}^{m-1}2t_j
\le T(W_{m-1})\sum_{i<m}u_i$.
Combining, $C(S_{\psi_m}W_{m-1})$ lies between $\tfrac12$ and $2$ times
$T(W_{m-1})\sum_{i=1}^{m}u_i$, which is \eqref{eq:nonoracle-theta}.
\end{proof}

\begin{remark}[Cancellation across the full run]
\label{rem:per-iterate}
The count \eqref{eq:nonoracle-theta} is per iterate, and the same
cancellation extends across the full run. Consecutive outer iterates
meet as $W_{m-1}S_{\psi_m}W_{m-1}$, the pattern of
Lemma~\ref{lem:diffuser-cancellation}, so only the boundary runs of the
first and last iterates survive. The first may be dropped, $\ket{\psi}$
being an eigenstate of every diffuser. The iterates of the cascade
cancel likewise, and contribute less than
the cost of two outer iterates. The total non-oracle cost is therefore
\eqref{eq:nonoracle-theta} with $T(W_{m-1})$ replaced by the total
oracle count.
\end{remark}

The weight $u_i$ compares the cost of the level-$i$ diffuser with the
repetitions accumulated beneath it, so the sum in
\eqref{eq:nonoracle-theta} stays bounded whenever the growth in cost
from one level to the next is offset by the iteration count at the
lower level.

\begin{corollary}[Geometric decay]
\label{cor:geometric-decay}
Suppose there is a constant $\bar r<1$ such that
\begin{equation}
\label{eq:geometric-condition}
    \frac{1}{\bar r}\cdot\frac{c(S_{\psi_{i+1}})}{2\,c(S_{\psi_i})}
    \le t_i,
    \qquad i<m.
\end{equation}
Then $u_{i+1}\le\bar r\,u_i$, and
\begin{equation}
\label{eq:geometric-bound}
    C(S_{\psi_m}W_{m-1}) = O\big(T(W_{m-1})\,c(S_{\psi_1})\big).
\end{equation}
\end{corollary}
\begin{proof}
Condition~\eqref{eq:geometric-condition} gives
$u_{i+1}/u_i=c(S_{\psi_{i+1}})/\big(2t_i\,c(S_{\psi_i})\big)\le\bar r$,
so $\sum_{i=1}^{m}u_i\le u_1/(1-\bar r)$ with $u_1=c(S_{\psi_1})$, and
Lemma~\ref{lem:nonoracle-cost} gives \eqref{eq:geometric-bound}.
\end{proof}

Condition~\eqref{eq:geometric-condition} bounds each $t_i$ from below,
while the schedule condition of Theorem~\ref{thm:query-optimal}
bounds it from above, yielding a schedule for $t_i$ that is optimal in both query complexity and non-oracle operator cost.

\subsection{Bounding Both Complexities}
\label{sec:window}
The two complexities constrain the schedule from opposite sides. The
oracle count of Theorem~\ref{thm:query-optimal} bounds each $t_i$ from
above, the non-oracle cost of Corollary~\ref{cor:geometric-decay}
bounds it from below, and a schedule is optimal in both cost measures
whenever every $t_i$ lies between the two. 

\begin{corollary}[Iteration-count window]
\label{cor:window}
Suppose there is a constant $\bar r<1$ such that the schedule satisfies
\begin{equation}
\label{eq:window}
    \frac{1}{\bar r}\cdot\frac{c(S_{\psi_{i+1}})}{2\,c(S_{\psi_i})}
    \;\le\; t_i \;\le\;
    \min\!\Big(\frac{\bar a}{2\sin\theta_{i+1}},\,
    \frac{\pi}{4\gamma_i}\Big),
    \qquad i<m,
\end{equation}
with $\bar a<1$ the constant of Theorem~\ref{thm:query-optimal}. Then
the oracle complexity is $\Theta(1/\sin\theta)$ and the non-oracle cost
of each iterate is $\Theta\big(T(W_{m-1})\,c(S_{\psi_1})\big)$.
\end{corollary}
\begin{proof}
The right-hand bound of~\eqref{eq:window} is, in both of its terms, the
schedule condition of Theorem~\ref{thm:query-optimal}, under which the
oracle complexity is $\Theta(1/\sin\theta)$. The left-hand bound is
condition~\eqref{eq:geometric-condition}, so the weights satisfy
$u_1\le\sum_{i=1}^{m}u_i\le u_1/(1-\bar r)$ with $u_1=c(S_{\psi_1})$,
and Lemma~\ref{lem:nonoracle-cost} gives non-oracle cost
$\Theta\big(T(W_{m-1})\,c(S_{\psi_1})\big)$.
\end{proof}

By Remark~\ref{rem:per-iterate}, the same count extends
across the full run, so the total non-oracle cost is
$\Theta\big(c(S_{\psi_1})/\sin\theta\big)$. Whenever the innermost
diffuser costs $O(1)$ in the chosen cost measure, both complexities are
therefore of the order $1/\sin\theta$, and any schedule
satisfying~\eqref{eq:window} is simultaneously optimal in oracle and
non-oracle cost. 

The analysis in this section is deliberately independent of how the
construction is implemented. The oracle complexity is fixed by the
local overlaps and the schedule, while the non-oracle cost is
captured by the abstract function $c(S_{\psi_i})$. This separates
the recursive dynamics from the physical cost of implementing the
diffusion operators. Section~\ref{sec:non-oracle} shows how the schedule can satisfy both requirements, with any schedule satisfying~\eqref{eq:window} simultaneously controlling the oracle and non-oracle costs. To apply
the construction in a particular search setting, it is therefore
enough to identify a suitable decomposition and substitute the
relevant diffuser costs. Section~\ref{sec:applications} does this
for unstructured search and spatial search on the $d$-dimensional
grid, where the non-oracle cost becomes, respectively, the
elementary gate count and the travel time of the walker.

\section{Instantiations}
\label{sec:applications}
We instantiate the construction in two settings, unstructured search
for a single target over $n$ qubits (Section~\ref{sec:combinatorial})
and spatial search for a single marked vertex on the $d$-dimensional
grid (Section~\ref{sec:spatial}). In both instances, the
initial and target states factorise over a natural partition of the search space, so our
algorithm applies and the analysis resolves to evaluating the schedule conditions of Corollary~\ref{cor:window} at the setting's local overlaps and
diffuser costs. For unstructured search our evaluation attains the
simultaneously optimal $\Theta(\sqrt{N})$ oracle and non-oracle gate
counts of Bria\'{n}ski et al.~\cite{hardwareGrover2}, and for the grid it
recovers the $O(\sqrt{N})$ time of Aaronson and
Ambainis~\cite{scott_paper} for $d\ge3$ and their
$O\big(\sqrt{N}\,(\log N)^{3/2}\big)$ at $d=2$.

\subsection{Unstructured Search}
\label{sec:combinatorial}

We consider unstructured search for a single target $x$ among $N=2^n$ values,
encoded in the computational basis of $n$ qubits so that the search space
is the Hilbert space $\mathcal{H}=(\mathbb{C}^2)^{\otimes n}$. A partition
of the qubits into $m=n/s$ blocks of $s$ is then a factorisation
$\mathcal{H}=\bigotimes_{i=1}^{m}\mathcal{H}_i$ into registers of
dimension $2^s$, a product partition in the sense of
Section~\ref{subsec:definitions}. Over it the target
$\ket{x}=\bigotimes_{i=1}^{m}\ket{x_i}$ is the product of its blocks and
the uniform superposition is the product of the local uniform
superpositions $\ket{+_i}$, so both states factorise and the recursion
applies as given.

Each block contains a local overlap
\begin{equation}
\label{eq:comb-overlap}
    \sin\theta_i = \braket{x_i|+_i} = 2^{-s/2},
\end{equation}
giving global overlap $\sin\theta=\prod_{i=1}^m\sin\theta_i=1/\sqrt{N}$, and
the level-$i$ diffuser $S_{\psi_i}$ is a reflection about the uniform
superposition over the lowest $i$ blocks. Acting on $is$ qubits, it is
realisable in $c(S_{\psi_i})=\Theta(is)$ elementary
gates~\cite{barenco, mcx1}.

\begin{theorem}[Unstructured search at constant block size]
\label{thm:combinatorial}
Let $N=2^n$, let the register be partitioned into $m=n/s$ blocks of
$s\ge3$ qubits, and let each partial diffuser be realised in
$c(S_{\psi_i})=\Theta(is)$ elementary gates. Then, at the minimal
schedule $t_i=1$, the construction prepares $\ket{x}$ with unit
probability using
\begin{equation}
\label{eq:comb-counts}
    T_{\mathrm{total}} = \big(1+\Theta(2^{-s})\big)\,\frac{\pi}{4}\sqrt{N}
    \quad\text{oracle calls and}\quad
    C_{\mathrm{non\text{-}oracle}} = \Theta\big(s\sqrt{N}\big)
    \quad\text{non-oracle gates}.
\end{equation}
\end{theorem}
\begin{proof}
We verify the schedule conditions of Theorem~\ref{thm:query-optimal} and
Corollary~\ref{cor:geometric-decay} in turn. For the first,
by~\eqref{eq:comb-overlap} the bound $t_i\le\bar a/(2\sin\theta_{i+1})
=\bar a\,2^{\,s/2-1}$ of~\eqref{eq:ti-explicit} admits $t_i=1$ for every
$s\ge3$ and any constant $\bar a\in(1/\sqrt{2},1)$, and its second bound
is met as well, since $\gamma_i\le\theta_i=\arcsin 2^{-s/2}<\pi/4$ gives
$\pi/(4\gamma_i)>1$. The overlaps $\theta_i=\arcsin 2^{-s/2}$ being
bounded, Theorem~\ref{thm:query-optimal} applies, and the oracle count is
$\Theta(1/\sin\theta)=\Theta(\sqrt{N})$. For the second, condition~\eqref{eq:geometric-condition} states
$t_i\ge c(S_{\psi_{i+1}})/\big(2\bar r\,c(S_{\psi_i})\big)
=(i+1)/(2\bar r\,i)$ up to implementation constants, which at
$\bar r=\tfrac34$ the schedule $t_i=1$ meets at every level $i\ge2$. Only the first level fails the condition, at ratio $u_2/u_1=1$, and a
single flat level costs at most a constant. The weights decay
geometrically from $u_2$ onwards, so
$\sum_{i=1}^{m}u_i\le 2\sum_{i=2}^{m}u_i\le 2u_2/(1-\bar r)=\Theta(s)$,
which Lemma~\ref{lem:nonoracle-cost} with Remark~\ref{rem:per-iterate}
forms the non-oracle count of~\eqref{eq:comb-counts}.

For a uniform partition with $\theta_i=\vartheta=\arcsin 2^{-s/2}$
and $t_i=1$, Corollary~\ref{cor:epsilon} applies with
$a=2\sin\vartheta=2^{1-s/2}\le2^{-1/2}$ and bounds the angle shortfall
by $\epsilon_s=\Theta(a^2)=\Theta(2^{-s})$, independent of $N$ and with
$\epsilon_s\le\tfrac12$ for every $s\ge3$. The total~\eqref{eq:explicit-total} then has leading term
$(1+2\epsilon_s)\tfrac{\pi}{4}\sqrt{N}$, Grover's leading
constant~\cite{bbht} up to an excess shrinking geometrically in $s$
and minimal at $t_i=1$. Its additive term is
$3\prod_{i<m}2t_i=\tfrac32 N^{1/s}$, which at constant $s\ge3$ is
$o\big(2^{-s}\sqrt{N}\big)$ and is absorbed by the excess
of~\eqref{eq:comb-counts}. Unit probability follows from
Theorem~\ref{thm:cascade}, whose two conditions,
$\lvert\gamma_k\rvert\le\theta_k<\pi/3$ and
$2t_k\gamma_k=2\gamma_k<\pi/2$, hold at every level.
\end{proof}

\begin{remark}[Robustness to the diffuser implementation]
\label{rem:diffuser-robustness}
The $\Theta(\sqrt{N})$ scaling does not require the partial
diffusers to have a linear gate count. Suppose instead that
$c(S_{\psi_i})=O\big((is)^p\big)$ for some fixed $p$, as may occur
when limited connectivity introduces a routing overhead. At the
minimal schedule $t_i=1$, the corresponding weights satisfy
\begin{align*}
    u_i &= O\!\left(\frac{(is)^p}{2^{i-1}}\right),
    \\
    \intertext{and hence}
    \sum_{i=1}^{m}u_i
    &= O\!\left(s^p\sum_{i=1}^{\infty}\frac{i^p}{2^{i-1}}\right)
    = O_p(s^p).
\end{align*}
Thus, for any fixed block size $s\ge3$ and fixed $p$, the
non-oracle cost remains $\Theta(\sqrt{N})$ in $N$, while the oracle
count is unchanged. For fixed block size, the $\Theta(\sqrt{N})$ scaling remains with a polynomial diffuser overhead.
\end{remark}

At $s=2$, $2t_i\sin\theta_{i+1}=1$ for $t_i = 1$, and by
Lemma~\ref{lem:boundary} the cost over the $m=n/2$ levels is
$\Theta\big(\sqrt{N}\sqrt{\log N}\big)$.

Equation~\eqref{eq:comb-counts} lies on the trade-off noted by Bria\'{n}ski
et al.~\cite{hardwareGrover2}. Their lower bound states that any
algorithm using $O\big(\log(1/\epsilon)\sqrt{N}\big)$ non-oracle gates
requires at least $(1+\epsilon)\tfrac{\pi}{4}\sqrt{N}$ oracle calls,
and with $\log(1/\epsilon_s)=\Theta(s)$ these are the results of
Theorem~\ref{thm:combinatorial}. Consequently, the block size traces
the trade-off directly, with each additional block qubit halving the
oracle excess at the cost of $\Theta(\sqrt{N})$ further gates.

\begin{remark}[Unequal blocks]
\label{rem:unequal-blocks}
While we use a uniform partition here, the construction admits unequal
block sizes. Any partition into blocks of $n_i\ge3$ qubits with
non-increasing, bounded sizes satisfies the schedule conditions of
Theorem~\ref{thm:query-optimal} and
Corollary~\ref{cor:geometric-decay} at $t_i=1$, and the counts of
Theorem~\ref{thm:combinatorial} hold with constants determined by
Theorem~\ref{thm:query-optimal} at the partition's overlaps
$\sin\theta_i=2^{-n_i/2}$.
\end{remark}

For unstructured search, a schedule can therefore be devised that is
optimal in both oracle complexity and non-oracle gate count,
reaching the same bounds as Bria\'{n}ski et al.~\cite{hardwareGrover2}. The
two constructions, however, differ in what the recursion determines.
Their approach describes a fixed procedure, recursively nesting a
single amplification step per level over a partition of the
register, where the amplitude gain at each level is determined by
the register size at that level. The amplitude concentrated in the
target is computed through a product recurrence in the block sizes, in which the accumulated loss stays bounded away
from zero only when the sizes $k_j$ grow with the level, $k_j=(x+1)j$ over
$m=\Theta(\sqrt{n/x})$ levels. At constant sizes $k\ge3$ that loss
decays exponentially in $n$ over the $n/k$ levels, and at $k=2$ it
vanishes entirely but the nesting itself makes $O(3^{n/2})$ oracle
calls, so no constant size attains $\Theta(\sqrt{N})$.

Here the reflections are nested with a tunable intermediate schedule
$\{t_i\}$, and the optimal counts are attained for any schedule
satisfying the conditions of Corollary~\ref{cor:window}, although
the minimal $t_i=1$ gives the smallest oracle excess. The state is
tracked exactly rather than the amplitude alone, and the loss at a
level is second order in the rotation that level applies, so the
losses converge at any constant $s\ge3$. The partition therefore
does not need to grow with the register, and the depth $m=n/s$ can be linear
in the number of qubits.

\subsection{Spatial Search}
\label{sec:spatial}
In this subsection, we consider spatial search for a single marked
vertex on a $d$-dimensional grid of $N$ vertices with side length
$L=N^{1/d}$, in the model of Aaronson and
Ambainis~\cite{scott_paper}. Where the travel cost in unstructured
search is negligible, in spatial search the database is spread across
a physical region and moving amplitude between distant sites takes
time. A quantum walker representing the state of the algorithm
searches the grid, carrying a vertex register $\ket{v}$ and a
workspace register $\ket{z}$. An oracle step writes the state of the
current vertex into the workspace,
$\ket{v,z}\mapsto\ket{v,z\oplus x_v}$, from which the oracle $S_x$
follows by phase kickback. Every other step moves amplitude only
along the edges of the grid, which join vertices differing by one in
a single coordinate. The cost of an algorithm is its number of steps,
so an operator correlating vertices at distance $\ell$ requires
$\Omega(\ell)$ of them.

Run plainly in this model, Grover search provides no speedup on the
$d=2$ grid. Its diffuser correlates amplitude across the whole grid,
so each of the $\Theta(\sqrt{N})$ iterates costs $\Theta(\sqrt{N})$
steps and the total is $\Theta(N)$~\cite{benioffRobot}. Aaronson and
Ambainis~\cite{scott_paper} recovered the speedup by recursively
subdividing the grid and partially amplifying within each level, in
$O(\sqrt{N})$ steps for $d\ge3$ and
$O\big(\sqrt{N}(\log N)^{3/2}\big)$ for $d=2$.

The Hilbert space of the walker encodes its position, and the vertex
register factorises along the axes as
$\ket{v}=\ket{a_1}\ket{a_2}\cdots\ket{a_d}$, each axis register
$\ket{a_r}$ spanned by the positions $\{\ket{0},\ldots,\ket{L-1}\}$. A
recursive subdivision of the grid is then a tensor-product partition
of this space, formed by writing the axes in mixed radix. Let the
$i$-th digit carry its own base $b_i$ with $L=\prod_{i=1}^{m}b_i$, the
same bases in the same places on every axis,
\begin{equation}
\label{eq:spatial-axis}
    \ket{a_r} = \ket{c_{r,m}}\ket{c_{r,m-1}}\cdots\ket{c_{r,1}},
    \qquad c_{r,i}\in\{0,\ldots,b_i-1\}.
\end{equation}
The digits partition the grid geometrically. Varying the lowest digit
of each axis and fixing the rest, the axes range over a sub-cube of
side $b_1$, the position within it indexed by the lowest digits and
the sub-cube itself by the rest. Likewise at each level $i$, the
digits above the $i$-th single out one cube of side
$\ell_i=\prod_{j\le i}b_j$, and the grid is divided by such cubes at
every level. Regrouping the digit factors by level rather than by
axis, the vertex register reads
\begin{equation}
\label{eq:spatial-regroup}
    \ket{v} \;=\; \bigotimes_{i=m}^{1}\ket{c_i},
    \qquad
    \ket{c_i} := \ket{c_{1,i}}\ket{c_{2,i}}\cdots\ket{c_{d,i}}
    \;\in\; \mathcal{H}_i \cong \big(\mathbb{C}^{b_i}\big)^{\otimes d},
\end{equation}
so that register $i$ holds one base-$b_i$ digit of every axis, in a
fixed axis order, and $\dim\mathcal{H}_i=b_i^{\,d}$. The rewriting is a relabelling of the walker's basis states, and a
step of the walker still changes a single axis by one. Within a cell
such a step changes only the lowest digit, and a step across a cell
boundary changes the digits above as well. Figure~\ref{fig:spatial-partition} illustrates it on a $16\times16$
grid.

Both the initial state and the target state factorise over this
partition of $\mathcal{H}$. The marked vertex can be written as a
product of its digits, and the uniform superposition as a product of the uniform superpositions over each digit register. The conditions of
Section~\ref{subsec:definitions} are therefore satisfied, and the
construction applies directly to the $m$-level partition.

The local overlaps follow immediately from the relabelling. At level
$i$, the register $\mathcal{H}_i$ distinguishes the $b_i^{\,d}$ sub-cubes
of side $\ell_{i-1}$ contained within one of side $\ell_i$, and the
target digit selects one of these values. Hence
\begin{equation}
\label{eq:spatial-overlap}
    \sin\theta_i = \braket{x_i|\psi_i} = \frac{1}{b_i^{\,d/2}},
    \qquad
    \prod_{i=1}^{m}\sin\theta_i = L^{-d/2} = \frac{1}{\sqrt{N}}.
\end{equation}

Under the relabelling, each partial diffuser acts independently on the
sub-cubes at its level. Recall that
\begin{equation}
    S_{\psi_i}
    = \mathbb{I}_{m\ldots i+1}\otimes
    \bigl(\mathbb{I}_{i\ldots 1}
    - 2\ket{\psi_{i\ldots 1}}\!\bra{\psi_{i\ldots 1}}\bigr).
\end{equation}
Fixing the registers above level $i$ selects one block in which the
lowest $i$ registers span a sub-cube $C$ of side $\ell_i$. Within this
block, $\ket{\psi_{i\ldots 1}}$ is the uniform state $\ket{s_C}$ over the
sub-cube. Hence
\begin{equation}
    S_{\psi_i}
    =
    \bigoplus_C
    \left(I_C-2\ket{s_C}\!\bra{s_C}\right).
    \label{eq:spatial-diffuser}
\end{equation}

To implement these reflections, let $U_i$ map a fixed corner of each
sub-cube to its uniform state, and let $Z_i$ apply a phase of $-1$ to
that corner in every sub-cube. Then
\begin{equation}
    S_{\psi_i}=U_iZ_iU_i^\dagger.
\end{equation}
In~\cite{scott_paper} it is shown that an arbitrary known
state on a graph can be prepared from a fixed vertex by fanning out
along a minimum-height spanning tree. For a $d$-dimensional sub-cube
of side $\ell_i$ rooted at a corner, this requires
$d(\ell_i-1)$ $C$-local steps. This count is also necessary, since
$\ket{s_C}$ has non-zero amplitude at the opposite corner, whose graph
distance from the root is $d(\ell_i-1)$. Thus
\begin{equation}
    c(U_i)=d(\ell_i-1).
\end{equation}
The corresponding transformations act independently within the
distinct sub-cubes at level $i$ and can therefore be applied in
parallel without transferring amplitude between them. They also act
trivially on the walker's workspace register. Since $Z_i$ requires
one step and $U_i^\dagger$ has the same cost as $U_i$, this
implementation gives
\begin{equation}
    c(S_{\psi_i})
    =2d(\ell_i-1)+1
    =\Theta(d\ell_i).
    \label{eq:diffuser-cost}
\end{equation}
The phase-generalised operator $S_{\psi_i}(\alpha)$ is obtained by
replacing the phase $-1$ in $Z_i$ by $e^{i\alpha}$, without changing
this cost. 

Since $\ell_{i+1}=b_{i+1}\ell_i$, the ratio between successive
diffuser costs is asymptotically $b_{i+1}/2$. Substituting this and
$\sin\theta_i=1/b_i^{d/2}$ into the schedule conditions of
Corollary~\ref{cor:window} gives, up to constant factors,
\begin{equation}
\label{eq:spatial-schedule}
    \frac{b_{i+1}}{2}
    \;<\;
    t_i
    \;<\;
    \frac{b_{i+1}^{\,d/2}}{2}.
\end{equation}
The lower bound keeps the non-oracle cost of the same order as the
oracle count, while the upper bound preserves the optimal oracle
complexity. For $d\ge3$ and bases $b_i\ge3$, this interval contains an
integer and therefore admits a valid schedule $\{t_i\}$. The initial state, the uniform superposition
over all vertices of the graph, has a preparation cost of $N^{1/d}$ through the same fanout approach as the outermost diffuser operator $S_{\psi_m}$.

\begin{theorem}[Spatial search, $d\ge3$]
\label{thm:spatial-dge3}
Let $d\ge3$, and let the grid be subdivided with bases $b_i$
bounded above by a constant and satisfying $b_i\ge3$, or $b_i\ge2$
for $d\ge5$. Choose
\begin{equation}
t_i=\left\lceil\frac{b_{i+1}}{2\bar r}\right\rceil
\end{equation}
for a constant $\bar r<1$ such that the schedule satisfies the
iteration-count window~\eqref{eq:window}, and suppose that
$w_i=2t_i\gamma_i\le\pi/2$ at every level. Then the construction
prepares the marked vertex with unit probability in
\begin{equation}
T=\Theta(\sqrt{N})
\end{equation}
steps.
\end{theorem}

\begin{proof}
Under the stated hypotheses, the schedule satisfies
condition~\eqref{eq:window} of Corollary~\ref{cor:window} with
constants $\bar a,\bar r<1$. The bounded bases keep the
overlaps~\eqref{eq:spatial-overlap} inside a fixed
$[\theta_{\min},\theta_{\max}]$. Corollary~\ref{cor:window}, together
with Remark~\ref{rem:per-iterate}, therefore gives
$\Theta(1/\sin\theta)=\Theta(\sqrt{N})$ oracle calls and travel cost
of the same order, since $c(S_{\psi_1})=\Theta(b_1)=\Theta(1)$.

The assumed bound $w_i\le\pi/2$ holds at every level, while
$b_i\ge2$ gives $\theta_i\le\pi/6<\pi/3$. The conditions of
Theorem~\ref{thm:cascade} are therefore satisfied, and the marked
vertex is prepared with unit probability.
\end{proof}

\begin{remark}[Uniform subdivisions]
\label{rem:uniform-spatial}
The hypotheses of Theorem~\ref{thm:spatial-dge3} include uniform
subdivisions $b_i=b$ whenever the corresponding constant schedule
satisfies the stated schedule conditions. For example, taking
$b=3$ and $t_i=2$ gives an admissible fixed-rate subdivision for
$d=3$ and $d=4$, while $b=2$ and $t_i=2$ does so for $d\ge5$.
More generally, any fixed satisfiable choice of $b$ and $t$ gives
$m=\Theta(\log N)$ while retaining the $\Theta(\sqrt{N})$ running
time. Grid sizes not admitting an exact subdivision can be padded to
the next compatible side length, increasing $N$ by only a constant
factor and therefore leaving the asymptotic running time unchanged.
\end{remark}

For $d\ge3$ the construction therefore reaches $\Theta(\sqrt{N})$
steps. This is the optimal lower bound, as $\Omega(\sqrt{N})$
oracle calls are necessary for any unstructured search problem~\cite{Optimal} and a $T$-step algorithm
makes at most $T$ of them. For $d=2$, however, the two sides of~\eqref{eq:spatial-schedule} are equal,
and no schedule satisfies the conditions of
Corollary~\ref{cor:window}.

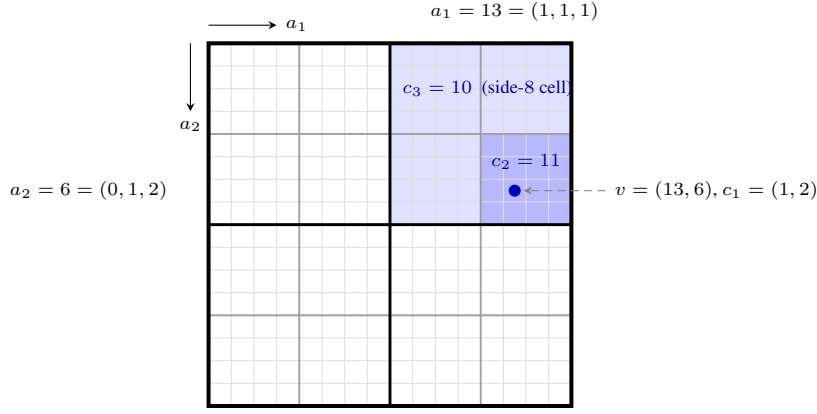
\begin{figure}[t]
\centering
\begin{tikzpicture}[scale=0.30]
\fill[blue!12] ( 8, 8) rectangle (16,16);
\fill[blue!28] (12, 8) rectangle (16,12);
\draw[gray!25, line width=0.2pt] (0,0) grid[step=1] (16,16);
\draw[gray!75, line width=0.7pt] (0,0) grid[step=4] (16,16);
\draw[black, line width=1.1pt] (0,0) grid[step=8] (16,16);
\draw[black, line width=1.5pt] (0,0) rectangle (16,16);
\node[circle, fill=blue!75!black, inner sep=1.6pt] (v) at (13.5,9.5) {};
\draw[->, >=stealth] (0,16.8) -- (3,16.8) node[right, font=\scriptsize] {$a_1$};
\draw[->, >=stealth] (-0.8,16) -- (-0.8,13) node[below, font=\scriptsize] {$a_2$};
\node[font=\scriptsize, anchor=south] at (13.5,16.6) {$a_1=13=(1,1,1)$};
\node[font=\scriptsize, anchor=east]  at (-1.4,9.5)  {$a_2=6=(0,1,2)$};

\node[font=\scriptsize, blue!60!black] at (12.3,14) {$c_3=10$ \;(side-$8$ cell)};
\node[font=\scriptsize, blue!60!black] at (14.0,10.8) {$c_2=11$};
\draw[densely dashed, gray, ->, >=stealth] (17.5,9.5) -- (13.9,9.5)
   node[pos=0, right, black, font=\scriptsize]
   {$v=(13,6)$, $c_1=(1,2)$};

\end{tikzpicture}
\caption{The relabelling of Eqs.~\eqref{eq:spatial-axis}
and~\eqref{eq:spatial-regroup} on a $16\times16$ grid, with bases
$(b_3,b_2,b_1)=(2,2,4)$ and the origin at the top-left corner. Register
$i$ holds digit $i$ of each axis, and fixing the registers above
level $i$ selects the sub-cube of side $\ell_i$ containing $v$, shaded
darker at each level.}
\label{fig:spatial-partition}
\Description{A sixteen by sixteen square grid with its origin at the
top-left corner, ruled at three nested scales: thin lines around every
single vertex, medium lines every four vertices, and heavy lines every
eight vertices, corresponding to the bases four, two and two read from the
lowest digit upwards. The marked vertex, at coordinates thirteen and six,
is drawn as a filled dot. The side-eight cell containing it is lightly
shaded and labelled by its level-three digit, and the side-four cell within
that is shaded darker and labelled by its level-two digit; the position
within the smallest cell is the level-one digit. The horizontal axis
coordinate thirteen is annotated with its three digits one, one, one, and
the vertical axis coordinate six with its digits zero, one, two.}
\end{figure}

\begin{theorem}[Spatial search, $d=2$]
\label{thm:spatial-d2}
On the two-dimensional grid, let the subdivision have bases
$b_i\ge2$ with $\prod_{i=1}^{m}b_i=\sqrt{N}$, and choose the boundary
schedule $t_i=b_{i+1}/2$. Suppose that the schedule contains integer values and
satisfies $w_i=2t_i\gamma_i\le\pi/2$ for every $i<m$. Then
\begin{gather}
    \label{eq:spatial-d2}
    T=\Theta\Big(
    \sqrt{N}\,m\,
    \sqrt{\textstyle\sum_{i=1}^{m}b_i^{2}}
    \Big).
    \\
    \intertext{Among such admissible boundary schedules,}
    T=\Omega\big(
    \sqrt{N}\,(\log N)^{3/2}
    \big),
\end{gather}
with this scaling attained by admissible subdivisions with
$m=\Theta(\log N)$ and bounded bases.

Under the stated conditions, the construction prepares the marked
vertex with unit probability.
\end{theorem}

\begin{proof}
At the boundary, $\sin\theta_i=1/b_i$ and, since
$w_i\le\pi/2$ for every $i<m$,
\begin{gather}
    \gamma_m
    = \Theta\!\left(\frac{1}{\sqrt{\sum_{i=1}^{m}b_i^2}}\right)
    \\
    \intertext{by~\eqref{eq:boundary-decay}. Also,
    $\prod_{i<m}2t_i=\prod_{i=2}^{m}b_i=\sqrt{N}/b_1$, so the oracle
    count is}
    \Theta\!\left(\frac{\sqrt{N}}{b_1}\sqrt{\sum_{i=1}^{m}b_i^2}\right).
    \\
    \intertext{At the boundary, the ratio in~\eqref{eq:geometric-condition}
    is $1$ at every level, so Lemma~\ref{lem:nonoracle-cost} and
    Remark~\ref{rem:per-iterate} contribute a further factor
    $\Theta(mb_1)$. The factors of $b_1$ cancel, giving}
    T = \Theta\!\left(\sqrt{N}\,m\,\sqrt{\sum_{i=1}^{m}b_i^2}\right),
\end{gather}
which proves~\eqref{eq:spatial-d2}.

We next minimise this expression over admissible boundary schedules.
Since $\prod_{i=1}^{m}b_i=\sqrt{N}$, we have
\begin{gather}
    \sum_{i=1}^{m}b_i^2 \ge mN^{1/m}.
    \\
    \intertext{Hence}
    T = \Omega\!\left(\sqrt{N}\,m^{3/2}N^{1/(2m)}\right).
    \\
    \intertext{The right-hand side is minimised for $m=\Theta(\log N)$,
    where $N^{1/(2m)}=\Theta(1)$, giving}
    T = \Omega\big(\sqrt{N}\,(\log N)^{3/2}\big).
\end{gather}
This scaling is attained whenever $m=\Theta(\log N)$ and the bases
remain bounded within the admissible class, since then
$\sum_i b_i^2=\Theta(m)$. If the bases are even, the boundary schedule is an integer; in particular, the binary subdivision $b_i=2$ gives
$t_i=1$ at every level. Finally, $b_i\ge2$ implies
$\theta_i\le\pi/6<\pi/3$. Together with the assumed condition
$w_i\le\pi/2$, the hypotheses of Theorem~\ref{thm:cascade} are
satisfied, and the marked vertex is prepared with unit probability.
\end{proof}

The two factors producing the $d=2$ overhead have
direct counterparts in the analysis of Aaronson and
Ambainis~\cite{scott_paper}. For bounded subdivisions, the decay of
the outer rotation angle gives an additional factor $\sqrt{m}$ in the
oracle count, corresponding to the $O(\sqrt{R})$
amplitude-amplification steps required in their construction when the
success probability falls as $\Omega(1/R)$. The non-oracle cost
contributes a further factor $m$, matching the factor $R$ accumulated
by their recursive local operations. At the minimising depth
$m=\Theta(\log N)$, these combine to give the same
$(\log N)^{3/2}$ overhead.

Despite the matching asymptotic costs, our spatial-search construction differs substantially from that of Aaronson and Ambainis~\cite{scott_paper}. There, the recursion runs through
amplitude amplification on subdivisions of the grid, and the success
probability of each level is bounded below, with the errors compounding multiplicatively.
For $d\ge3$ they are
absorbed by sub-cube sizes that grow doubly exponentially,
$n_{R-1}=\Theta(n_R^{\beta})$ with $\beta<1$, over
$\Theta(\log\log N)$ levels, a growth that also keeps their travel
cost geometric. For $d=2$ the diameter of the grid is
$\Theta(\sqrt{N})$, so only $O(1)$ amplification steps can be applied
at the end, meaning that the success probability must remain high
throughout the recursion, where amplification is least efficient.
Their recursion then runs at constant sub-cube sizes over
$\Theta(\log N)$ levels, costing one factor of $\log N$ for the
accumulated travel and a further $\sqrt{\log N}$ for the final amplification.

In our construction, the recursive subdivision is instead a
factorisation of the walker's Hilbert space. Each level reflects
about the image of the uniform superposition over its sub-cube under
the rotation that the level itself applies, so the level composes to
a rotation whose angle Section~\ref{sec:dynamics} tracks exactly in contrast to bounding the success probability of the subroutines during the recursion. The loss at a level is then second order in
its own rotation $w_i=2t_i\gamma_i$. For $d\ge3$, the schedule conditions~\eqref{eq:spatial-schedule}
make the $w_i$ decay geometrically, so $\sum_i w_i^2$ remains
bounded without increasing the sub-cube sizes, and the construction
attains $\Theta(\sqrt{N})$. For $d=2$ the
two sides of~\eqref{eq:spatial-schedule} coincide and no schedule
satisfies them. Here, the angle $w_i$ decays with the depth, and
$\gamma_m=\Theta(1/\sqrt{m})$ by Lemma~\ref{lem:boundary}. The
oracle count then gains a factor $\sqrt{m}$ and the diffusers a
further $m$, which at the minimising depth $m=\Theta(\log N)$
combine to $\Theta\big(\sqrt{N}(\log N)^{3/2}\big)$.

\section{Discussion}
\label{sec:discussion}

Building a quantum search algorithm around a decomposition of the
Hilbert space encoding the problem can yield advantages over
standard Grover search and amplitude amplification. We propose an approach
that recurses on this decomposed space through nested
reflection operators. When the initial and target states factorise
over the decomposition, the non-trivial dynamics at each level are
confined to a single two-dimensional invariant plane, and the
displacement of the initial state within that plane is determined
exactly by a scalar recurrence. An exact description of the state is
therefore available at any stage of the algorithm, simplifying the
analysis of recursive quantum search since properties of the
algorithm follow directly from the state's evolution instead of
through separate success-probability bounds at each level. This
makes it possible to reason about the recursion as a whole, rather
than controlling the accumulated error of a sequence of approximate
sub-searches. The resulting information can then be used when
implementing the construction within a particular search setting and
when optimising it for the costs relevant to that setting.


The Hilbert spaces of unstructured search for a single
target and spatial search for a single marked vertex on a
$d$-dimensional grid both admit decompositions for which our
construction applies. By abstracting the implementation cost of the
operators, we can derive the schedules under which the oracle
complexity and the cost of the non-oracle operators are
simultaneously minimised. The same recursive construction can
therefore be optimised against different physical costs without
changing its underlying analysis. It can be instantiated within a
setting by substituting the relevant operator costs, such as the elementary
gate count in unstructured search and the travel time of the walker
in spatial search. The resulting implementations attain the best
known costs for unstructured search and recover the established
recursive-search bounds for spatial search.

Furthermore, the analysis technique enabled by the recursion allows
us to establish efficient schedules that are not covered by the
corresponding prior analyses. In particular, fixed-size
decompositions attain the optimal oracle and non-oracle costs for
unstructured search and for spatial search with $d\ge3$. The
analyses of Bria\'{n}ski et al.~\cite{hardwareGrover2} and of
Aaronson and Ambainis~\cite{scott_paper}, respectively, operate within subdivision rates that grow with the recursion level. The construction
therefore attains these costs without requiring this growth, giving
greater freedom in choosing the decomposition and the operators used
for the recursion.

\subsection{Limitations and Future Work}
In applying and developing the scheme, several considerations emerged that help clarify its scope and implementation. These also point to natural directions in which the construction may be extended.

A central condition of the construction is that the initial and
target states factorise over the same product partition of the
Hilbert space. For the initial state, this factorisation allows the
partial diffusers to act only on part of the system while leaving the
remainder unchanged. The uniform superposition has this structure
naturally. An arbitrary initial state may instead be brought into
product form by augmenting the construction with a unitary that
disentangles it across the chosen partition. The corresponding
partial diffusers must then be conjugated by this transformation at
every level, so they differ from a reflection about the full initial
state only in the central reflection about the appropriate zero
state. Although the schedule conditions remain unchanged, the cost
of each diffuser can therefore approach that of the global
reflection, and the repeated conjugation can dominate the
non-oracle cost. Extending the construction to arbitrary initial
states while retaining partial diffusers that are cheap to implement
remains an open problem.

For the target state, factorisation is not required to construct the recursive sequence of reflections, but it is what confines the non-trivial dynamics at each level to a single invariant plane and guarantees the scalar recurrence for the rotation angle. This analytical simplification does not restrict the construction to a single marked element, since a target state supported on multiple basis states may itself factorise over the partition. A generic multiple-target state, however, may not have this structure, in which case the single-angle description would no longer apply as given. One possible extension would be to replace the scalar recurrence by a low-dimensional description associated with the Schmidt structure of the target across the partition. Determining when such a description remains closed under the recursion merits further investigation.

The construction also assumes that the local overlaps are known in
advance. Standard amplitude amplification requires the global
overlap to determine the optimal iteration count, whereas here the
outer rotation angle $\gamma_m$ and the phases used to resolve each
register are determined by all $m$ local overlaps. For unstructured
search with a single target, these values follow directly from the
dimensions of the subspaces. The full set of local overlaps
$\theta_i$ is required for deterministic preparation, but this
dependence can be relaxed when exact preparation is not required.
With the minimal schedule $t_i=1$,
Theorem~\ref{thm:query-optimal} gives
$\sin\gamma_m=\Theta(2^{m-1}\sin\theta)$, so the outer iteration
count can be estimated from the global overlap.
Remark~\ref{rem:no-cascade} further shows that the phase-tuned
corrections can be omitted with a failure probability quadratic in
the remaining angles, which the schedules considered here keep at
$O(\sin^2\theta_1)$. A natural extension would be a fixed-point
version of the recursion, analogous to fixed-point amplitude
amplification~\cite{FixedPointSearch, fixed-point}, that further
reduces the dependence on precise knowledge of the local overlaps.

More broadly, the recursion is not tied to its interpretation as a
search algorithm. Quantum algorithms with reflection- and
rotation-based dynamics similar to those of quantum search may admit
a similar decomposition of their Hilbert space and reflection
operators, allowing the recursive construction and analysis developed
here to be applied in those settings. Szegedy-type quantum
walks~\cite{szegedy} provide one possible direction, since their
walk operator is a product of reflections whose action decomposes
into two-dimensional invariant subspaces. A recursive partition of
the walk space may therefore allow the corresponding rotation angles
to be related across successive levels, in a manner analogous to the
angle recurrence derived here. A further possible connection is with
quantum singular value transformation~\cite{QSVT}, which applies
polynomial functions to the singular values of an operator encoded
within a unitary. In the present construction, repeated rotations
likewise induce a transformation of the singular value associated
with a pair of reflected subspaces, with the transformed value
determining the singular value at the next level through the angle
recurrence.
\section{Conclusion}
\label{sec:conclusion}

We have developed a novel decomposition of quantum search in which the
search operator is built recursively from reflections over a
decomposition of the underlying Hilbert space. When the initial and
target states factorise over a common product partition, the
non-trivial dynamics at each level are confined to a single invariant
plane. The corresponding rotation angle is propagated exactly through
the recursion by a scalar recurrence. The success probability
therefore follows directly from the accumulated rotation, with no
residual error from bounding the success of individual stages. This
allows the recursion to be treated as a whole and gives an exact
description of the state throughout the algorithm. Furthermore, we
show how the phases of the reflection operators can be adjusted to
prepare the target deterministically.

This description also allows the oracle and non-oracle costs of the
construction to be derived abstractly and then instantiated according
to the cost model of a particular search setting. For unstructured
search, the resulting algorithm attains the simultaneously optimal
$\Theta(\sqrt{N})$ oracle and non-oracle gate counts of
Bria\'{n}ski et al.~\cite{hardwareGrover2}. For spatial search on the
$d$-dimensional grid, it recovers the $O(\sqrt{N})$ running time for
$d\ge3$ and the $O\big(\sqrt{N}(\log N)^{3/2}\big)$ bound of
Aaronson and Ambainis~\cite{scott_paper} for $d=2$.

The construction also admits recursive decompositions not covered by
the corresponding prior analyses. In particular, for unstructured
search and spatial search with $d\ge3$, the optimal asymptotic costs
can be attained using a fixed rate of subdivision at every level,
whereas the corresponding prior analyses use subdivision rates that
increase with recursion depth. This gives greater freedom in how the
recursion is structured and in choosing decompositions suited to the
available operations. More broadly, the recursive construction and its exact angle-based analysis suggest a new approach to incorporating and analysing recursion in quantum algorithm design.

\printbibliography

\appendix
\section{Deferred Proofs}
\label{app:proofs}
Two proofs are collected here rather than given where their statements
appear. The first is the induction producing an orthonormal basis for
the level-$i$ rotation plane and fixing its orientation
(Lemma~\ref{lem:basis}, Section~\ref{subsec:planes}). The second is
the computation tracking the reciprocal squared angle along the
boundary schedule (Lemma~\ref{lem:boundary}). The sketches standing
in their place quote two identities, \eqref{eq:reflection-axis}
and~\eqref{eq:boundary-increment}, which later sections draw from the
proofs rather than from the statements.

\subsection{Orthonormal Basis for the Rotation Plane}
\label{app:basis-proof}
Throughout, $\ket{x_i^\perp} = \Pi_i^\perp\ket{\psi_i}/\cos\theta_i$ is
the unit vector of~\eqref{eq:perp-vector}, so that
$\braket{x_i^\perp|\psi_i} = \cos\theta_i$ by~\eqref{eq:perp-overlap}.
The reduced projectors $\tilde{P}_\psi, \tilde{P}_W$ are those
of~\eqref{eq:reduced-projectors}. By Lemma~\ref{lem:reduction},
$S_{\psi_i}$ and $W_{i-1}$ act through them within
$\mathcal{H}_i \otimes \cdots \otimes \mathcal{H}_1$ once the common
factor $\ket{x_{m\ldots i+1}}$ on the upper registers is suppressed.
\begin{proof}[Proof of Lemma~\ref{lem:basis}]
We induct on $i$. The base case $i = 1$ is the rotation by
$\pi - 2\theta_1$ sketched in Section~\ref{subsec:planes}, a direct
computation in $\mathrm{span}\{\ket{x_1}, \ket{\psi_1}\}$ orienting it
as in~\eqref{eq:base-evolution} with $\ket{\psi_1^\perp}$ as defined.

For the inductive step, fix $i \ge 2$, suppress the common factor
$\ket{x_{m\ldots i+1}}$ and work within
$\mathcal{H}_i \otimes \cdots \otimes \mathcal{H}_1$, where
$S_{\psi_i}$ and $W_{i-1}$ act through $\tilde{P}_\psi$ and
$\tilde{P}_W$. The plane is $\mathrm{span}\{\ket{\psi_{i\ldots 1}},
\tilde{P}_W\ket{\psi_{i\ldots 1}}\}$, so take
$\ket{e_1^{(i)}} = \ket{\psi_{i\ldots 1}}$ and obtain
$\ket{e_2^{(i)}}$ by Gram--Schmidt. We apply $\tilde{P}_W$ to
$\ket{\psi_i}\ket{\psi_{i-1\ldots 1}}$, using
$\braket{x_i|\psi_i} = \sin\theta_i$ and
$\braket{x_i^\perp|\psi_i} = \cos\theta_i$ together with the evolution
one level down, which is \eqref{eq:plane-evolution} at level $i-1$ by
the inductive hypothesis, or \eqref{eq:base-evolution} if $i = 2$,
whose sign $(-1)^{t_1}$ cancels since the evolved state enters
$\tilde{P}_W$ once in the ket and once in the bra. This gives
\begin{equation}
\begin{split}
    \tilde{P}_W\ket{\psi_i}\ket{\psi_{i-1\ldots 1}} = {}&
    \sin\theta_i\cos(2t_{i-1}\gamma_{i-1})\ket{x_i}\otimes
    \big(\cos(2t_{i-1}\gamma_{i-1})\ket{\psi_{i-1\ldots 1}}
    - \sin(2t_{i-1}\gamma_{i-1})\ket{\psi_{i-1\ldots 1}^\perp}\big) \\
    & + \cos\theta_i\ket{x_i^\perp}\ket{\psi_{i-1\ldots 1}}.
\end{split}
\label{eq:PWui-raw}
\end{equation}
Using $\cos^2 = 1 - \sin^2$, the $\ket{\psi_{i-1\ldots 1}}$-terms
regroup as $\ket{\psi_i}\ket{\psi_{i-1\ldots 1}}$, and the remainder
factors through
$\sin\gamma_i = \sin\theta_i\sin(2t_{i-1}\gamma_{i-1})$, so
\begin{equation}
    \tilde{P}_W\ket{\psi_i}\ket{\psi_{i-1\ldots 1}} =
    \ket{\psi_i}\ket{\psi_{i-1\ldots 1}}
    - \sin\gamma_i\ket{x_i}\otimes
    \Big[\sin(2t_{i-1}\gamma_{i-1})\ket{\psi_{i-1\ldots 1}}
    + \cos(2t_{i-1}\gamma_{i-1})\ket{\psi_{i-1\ldots 1}^\perp}\Big].
    \label{eq:PWui-compact}
\end{equation}
Subtracting the $\ket{\psi_{i\ldots 1}}$-component, whose coefficient
is $\cos^2\gamma_i$ by Lemma~\ref{lem:eigenvalue}, leaves
\begin{equation}
    \ket{w_i} = \sin\gamma_i\Big[\sin\gamma_i
    \ket{\psi_i}\ket{\psi_{i-1\ldots 1}} - \ket{x_i}\otimes
    \big(\sin(2t_{i-1}\gamma_{i-1})\ket{\psi_{i-1\ldots 1}}
    + \cos(2t_{i-1}\gamma_{i-1})\ket{\psi_{i-1\ldots 1}^\perp}\big)\Big],
\end{equation}
of norm $\lvert\sin\gamma_i\cos\gamma_i\rvert$. Dividing by the signed
quantity $\sin\gamma_i\cos\gamma_i$ rather than by that norm yields
$\ket{e_2^{(i)}} = \ket{\psi_{i\ldots 1}^\perp}$ as in~\eqref{eq:e2}.
The two prescriptions differ by a sign whenever $\gamma_i < 0$, and it
is the signed one that fixes the orientation of the plane, so that a
single convention makes the iterate rotate by $-2\gamma_i$, with the
signed angle of~\eqref{eq:engine-recurrence}, at every level. For
$\gamma_i = 0$ the plane is degenerate, $S_{\psi_i}W_{i-1}$ acts on it
as the identity, and~\eqref{eq:e2} still returns a unit vector
orthogonal to $\ket{\psi_{i\ldots 1}}$.

It remains to fix the sign of the rotation. The same expression gives
$\tilde{P}_W\ket{\psi_{i\ldots 1}} = \cos^2\gamma_i\ket{e_1^{(i)}} +
\sin\gamma_i\cos\gamma_i\ket{e_2^{(i)}}$, which
is~\eqref{eq:reflection-axis}. Since $W_{i-1}$ acts as
$\mathbb{I} - 2\tilde{P}_W$ here (Lemma~\ref{lem:reduction}), it maps
$\ket{e_1^{(i)}}$ to
$-\cos(2\gamma_i)\ket{e_1^{(i)}} - \sin(2\gamma_i)\ket{e_2^{(i)}}$,
and $S_{\psi_i} = \mathbb{I} - 2\ket{e_1^{(i)}}\bra{e_1^{(i)}}$ flips
the first component and fixes the second, leaving
$\cos(2\gamma_i)\ket{e_1^{(i)}} - \sin(2\gamma_i)\ket{e_2^{(i)}}$.
This is $\ket{e_1^{(i)}}$ rotated by $-2\gamma_i$, and iterating
gives~\eqref{eq:plane-evolution}.
\end{proof}
\subsection{Angle Decay at the Boundary Schedule}
\label{app:boundary-proof}
Here $w_i := 2t_i\gamma_i$ is the rotation of level $i$, held at or
below $\pi/2$ by hypothesis. The global overlap is
$\sin\theta = \prod_{i=1}^{m}\sin\theta_i$, as
in~\eqref{eq:overlaps}, and $T_{\mathrm{total}}$ is the total oracle
count~\eqref{eq:total-cost}.
\begin{proof}[Proof of Lemma~\ref{lem:boundary}]
At the boundary $\sin\theta_{i+1}=\gamma_i/w_i$, so the recurrence of
Lemma~\ref{lem:recursion} closes over the angles alone,
\begin{equation}
\label{eq:boundary-map}
    \sin\gamma_{i+1}=\sin\theta_{i+1}\sin w_i
    =\gamma_i\,\frac{\sin w_i}{w_i}.
\end{equation}
We bound the growth of $y_i:=1/\gamma_i^2$ level by level, using
$w_i\le\pi/2$ throughout. For the upper increment,
$\sin w\ge w(1-w^2/6)$ gives
$\gamma_{i+1}\ge\sin\gamma_{i+1}\ge\gamma_i(1-\tfrac{w_i^2}{6})$, and
since $w_i^2/6\le\pi^2/24$,
\begin{equation}
    y_{i+1}\ \le\ y_i\Big(1-\frac{w_i^2}{6}\Big)^{-2}
    \ \le\ y_i\Big(1+\frac{5w_i^2}{6}\Big),
    \qquad\text{so}\qquad
    y_{i+1}-y_i\ \le\ \frac{5}{6}\,\frac{w_i^2}{\gamma_i^2}
    \ =\ \frac{5}{6\sin^2\theta_{i+1}}.
\end{equation}
For the lower increment, $\sin w\le w(1-\tfrac{w^2}{7})$ on
$(0,\pi/2]$ bounds \eqref{eq:boundary-map} above, while $t_i\ge1$
forces $\sin\gamma_{i+1}\le\sin\theta_{i+1}\le\tfrac12$, on which
range $\arcsin z\le z(1+\tfrac{z^2}{5})$. Combining, with
$\sin^2\gamma_{i+1}\le\sin^2\theta_{i+1}\,w_i^2\le\tfrac{w_i^2}{4}$,
\begin{equation}
    \gamma_{i+1}\ \le\ \gamma_i\Big(1-\frac{w_i^2}{7}\Big)
    \Big(1+\frac{w_i^2}{20}\Big)
    \ \le\ \gamma_i\Big(1-\frac{w_i^2}{11}\Big),
    \qquad\text{so}\qquad
    y_{i+1}-y_i\ \ge\ \frac{2}{11}\,\frac{w_i^2}{\gamma_i^2}
    \ =\ \frac{2}{11\sin^2\theta_{i+1}}.
\end{equation}
Together the two increments are the
estimate~\eqref{eq:boundary-increment} quoted in the sketch, with
constants $\tfrac{2}{11}$ and $\tfrac{5}{6}$ independent of the
overlaps. Summing from $y_1=1/\theta_1^2$, with
$\theta/\sin\theta\in[1,\pi/2]$ making
$y_1=\Theta(\sin^{-2}\theta_1)$ uniformly,
\begin{equation}
    y_m=\Theta\Big(\textstyle\sum_{i=1}^{m}\sin^{-2}\theta_i\Big),
\end{equation}
which is~\eqref{eq:boundary-decay}. For overlaps inside a fixed
$[\theta_{\min},\theta_{\max}]$ every term is $\Theta(1)$, so
$y_m=\Theta(m)$ and $\gamma_m=\Theta(1/\sqrt{m})$. For the cost, the
boundary sets $\sin\theta_{i+1}=1/2t_i$ at every level, so
\begin{equation}
    \sin\theta=\sin\theta_1\prod_{i=1}^{m-1}\sin\theta_{i+1}
    =\sin\theta_1\prod_{i=1}^{m-1}\frac{1}{2t_i},
\end{equation}
giving $\prod_{i=1}^{m-1}2t_i=\Theta(1/\sin\theta)$. Substituting
this and $\gamma_m=\Theta(1/\sqrt{m})$
into~\eqref{eq:total-cost} yields
$T_{\mathrm{total}}=\Theta(\sqrt{m}/\sin\theta)$,
which is~\eqref{eq:boundary-cost}.
\end{proof}

\end{document}